\documentclass[a4paper,11pt]{article}%
\usepackage{amsfonts}
\usepackage{amssymb}
\usepackage{algo}
\usepackage[centertags]{amsmath}
\usepackage{graphicx}
\usepackage{rotating}
\usepackage{color}
\usepackage[left=2.05cm, right=2.05cm, top=2.0cm, bottom=2.2cm]{geometry}%
\providecommand{\U}[1]{\protect\rule{.1in}{.1in}}
\pdfoutput=1
\newtheorem{theorem}{Theorem}

\newtheorem{lemma}{Lemma}

\newenvironment{proof}[1][Proof]{\textbf{#1.} }{\  \rule{0.5em}{0.5em}}
\begin{document}

\title{Forecasting observables with particle filters: Any filter will do!\thanks{The
research has been supported by Australian Research Council Discovery Grants
DP150101728 and DP170100729.}}
\author{Patrick Leung\thanks{Department of Econometrics and Business Statistics,
Monash University. }, Catherine S. Forbes\thanks{Department of Econometrics
and Business Statistics, Monash University.}, Gael M. Martin\thanks{Department
of Econometrics and Business Statistics, Monash University. Corresponding
author; email: gael.martin@monash.edu; postal address: 20 Chancellors Walk,
Monash University, Clayton 3800, Australia.} and Brendan
McCabe\thanks{Management School, University of Liverpool.}}
\maketitle

\begin{abstract}
We investigate the impact of filter choice on forecast accuracy in state space
models. The filters are used both to estimate the posterior distribution of
the parameters, via a particle marginal Metropolis-Hastings (PMMH) algorithm,
and to produce draws from the filtered distribution of the final state.
Multiple filters are entertained, including two new data-driven methods.
Simulation exercises are used to document the performance of each PMMH
algorithm, in terms of computation time and the efficiency of the chain. We
then produce the forecast distributions for the one-step-ahead value of the
observed variable, using a fixed number of particles and Markov chain draws.
Despite distinct differences in efficiency, the filters yield virtually
identical forecasting accuracy, with this result holding under both correct
and incorrect specification of the model. This invariance of forecast
performance to the specification of the filter also characterizes an empirical
analysis of S\&P500 daily returns.

\bigskip

\bigskip


\end{abstract}

\baselineskip20pt

\section{Introduction}

Markov chain Monte Carlo (MCMC) schemes for state space models in which the
likelihood function is estimated using a particle filter, have expanded the
toolkit of the Bayesian statistician. The seminal work on \textit{particle}
MCMC (PMCMC) by Andrieu, Doucet and Holenstein (2010) in fact illustrates the
more general concept of \textit{pseudo-marginal} MCMC, in which insertion of
an unbiased estimator of the likelihood within a Metropolis-Hastings (MH)
algorithm is shown to yield the correct invariant distribution (Beaumont,
2003; Andrieu and Roberts, 2009). Subsequent work on PMCMC by Flury and
Shephard (2011)\ and Pitt, dos Santos Silva, Giordani and Kohn (2012) has
explored the interface between different filtering-based estimates of the
likelihood and the mixing properties of the resultant PMCMC algorithms in a
variety of settings, with the features of the true data generating process -
in particular the signal-to-noise ratio in the assumed state space model (SSM)
- playing a key role in the analysis. See Whiteley and Lee (2014), Del Moral,
Doucet, Jasra, Lee, Yau and Zhang (2015), Del Moral and Murray (2015),
Guarniero, Johansen and Lee (2017), Deligiannidis, Doucet and Pitt (2018),
Doucet and Lee (2018) and Quiroz, Tran, Villani and Kohn (2018, 2019) for a
range of more recent contributions to the area.

The aim of this paper is a very particular one: to explore the implications
for \textit{forecast accuracy} of using different particle marginal MH (PMMH)
algorithms. That is, we address the question of whether or not the specific
nature of the filter - used both to construct the likelihood estimate and to
draw from the filtered distribution of the final state - affects the forecast
distribution of the observed, or measured, random variable in the state space
model. This is, as far as we are aware, a matter that has not yet been
investigated, and is one of practical relevance to those researchers whose
primary goal is prediction, rather than inference \textit{per se}. The
performance of several well-established methods is explored, namely\textbf{
}the bootstrap particle filter (BPF) of Gordon, Salmond and Smith (1993), the
auxiliary particle filter (APF) of Pitt and Shephard (1999) and the unscented
particle\ filter (UPF) of van de Merwe, Doucet, de Freitas and Wan
(2000).\textbf{ }In order to broaden the scope of the investigation, we
introduce and include two new particle filters. Drawing on a particular
representation of the components in an SSM, as first highlighted in Ng,
Forbes, Martin and McCabe (2013), our first new filter provides a mechanism
for generating independent proposal draws using information on the current
data point only, and we use the term `data-driven particle filter' (DPF) to
refer to it as a consequence. The second filter is a modification of this
basic DPF - a so-called `unscented' DPF (UDPF) - which exploits unscented
transformations (Julier, Uhlmann and Durrant-Whyte, 1995; Julier and Uhlmann,
1997) in conjunction with the DPF mechanism to produce draws that are informed
by both the current observation and the previous state. The paper thus
includes a range of\textbf{ }particle\textbf{ }filters that, in varying ways,
allow for differential impact of the current observation and the past
(forecasted) information about the current state variable. Hence, the
conclusions we draw about forecast performance cannot be deemed to be unduly
influenced by focusing on too narrow a class of filter.

Using simulation, from a range of state space models, and under different
scenarios for the signal-to-noise ratio, we first document the performance of
each filter-specific PMMH algorithm, in terms of computation time and the
efficiency of the chain. We then produce the forecast distributions for the
one-step-ahead value of the observed random variable, documenting forecast
performance over a hold-out period. The key result is that, despite
differences in efficiency, the alternative filters yield virtually identical
forecasting accuracy, with this result holding under both correct and
incorrect specification of the true DGP.

To further substantiate this conclusion, we use alternative filters in an
empirical setting in which a relatively simple stochastic volatility
(SV)\textbf{\ }model is estimated using data from Standard and Poor's
Composite Index, denoted simply as the S\&P 500. Given that this data has
typically been modelled using a\textbf{ }much more complex process, it is
likely that the model is misspecified, and inference itself impacted by that
misspecification. Despite this fact we find, again, that the probabilistic
forecasts produced by the different filters are virtually identical and,
hence, yield the same degree of forecast accuracy.

In Section \ref{like_est} we begin with an outline of the role played by
particle filtering in likelihood estimation and implementation of a PMMH
scheme followed, in Section \ref{forecasting}, with a description of its role
in producing an estimate of the one-step-ahead forecast distribution for the
observed variable. In Section \ref{DPF_alg} we first provide a brief outline
of the existing, and well-known, filters that we include in our
investigations: the BPF, APF and UPF. The two new filters that we introduce,
the DPF and the UDPF, are then described. The computational benefit of using a
multiple matching of particles (Lin, Zhang, Cheng and Chen, 2005) in the
production of the likelihood estimate is explored in the context of the DPF,
and in Appendix A.3 it is established that the likelihood estimators resulting
from both new filters are unbiased.

Two different simulation exercises are conducted in Section \ref{sim_ex}%
.\textbf{ }The first investigates the relative computational burden of each of
several distinct filter types, along with the resulting impact on the mixing
properties of the corresponding Markov chain. This exercise is based on three
alternative state space models: i) the linear Gaussian model; ii) the
stochastic conditional duration model of Bauwens and Veredas (2004) (see also
Strickland, Forbes and Martin, 2006); and iii) the stochastic volatility model
of Taylor (1982) (see also Shephard, 2005). The alternative filter
types\textbf{ }are used to estimate the likelihood function within an adaptive
random walk-based MH algorithm. For each method we record both the `likelihood
computing time' associated with each filtering method - namely the
average\textbf{ }time taken to produce a likelihood estimate with a given
level of precision at some representative (vector) parameter value - and the
inefficiency factors associated with the resultant PMMH algorithm. In so doing
we follow the spirit of the exercise undertaken in Pitt\textit{ \textit{et
al.}} (2012), amongst others, in which a balance is achieved between
computational burden and the efficiency of the resultant Markov chain;
measuring as we do the time taken to produce a likelihood estimate that is
sufficiently accurate to yield an acceptable mixing rate in the chain.

In the second simulation exercise, forecasting performance is the focus. We
estimate the one-step-ahead forecast (or\textbf{ }predictive) distribution of
the observed variable, for each of the different filters and associated PMMH
algorithms, repeating the exercise (using expanding windows) over a hold-out
period, and assessing forecast accuracy using a logarithmic scoring rule
(Gneiting and Raftery, 2007). We first generate artificial data from a
stochastic volatility model, and produce the forecast distributions using the
correctly specified model. We then keep the forecast model the same, but
generate data from a stochastic volatility model with both price and
volatility jumps, in order to assess the impact on forecast performance of
model misspecification. To allow the documented differences in the performance
of the different filters to potentially affect forecast accuracy, in this
exercise we hold the number of particles, and the number of MCMC draws, fixed
for all PMMH methods. Despite this, the alternative filters yield virtually
identical forecasting accuracy, with this result holding under both correct
and incorrect specification of the true DGP. In Section \ref{emp_section}, an
empirical study is undertaken, in which the competing PMMH\ algorithms are
used to produce one-step-ahead forecast distributions from an SV model for the
S\&P 500 data. Again, despite restricting the different filters to operate
with the same number of particles, and the resulting Markov chains to have the
same number of iterations, the resulting forecast performance of the competing
methods is essentially equivalent. Section \ref{conclude} concludes.

\section{\textbf{PMMH and Forecasting}}

\subsection{Unbiased likelihood estimation and PMMH\textbf{\label{like_est}}}

In our context, an SSM describes the evolution of a latent state variable,
denoted by $x_{t}$, over discrete times $t=1,2,...$, according to the state
transition probability density function (pdf),
\begin{equation}
p\left(  x_{t+1}|x_{t},\theta\right)  , \label{state}%
\end{equation}
and with initial state probability given by $p\left(  x_{0}|\theta\right)  $,
where $\theta$ denotes a vector of unknown parameters. The observation in
period $t$, denoted by $y_{t}$, is modelled conditionally given the
contemporaneously indexed state variable via the conditional measurement
density
\begin{equation}
p\left(  y_{t}|x_{t},\theta\right)  . \label{meas}%
\end{equation}
Without loss of generality we assume that both $x_{t}$ and $y_{t}$ are scalar.

Typically, the complexity of the model is such that the likelihood function,%
\begin{equation}
L(\theta)=p(y_{1:T}|\theta)=p(y_{1}|\theta)\prod\limits_{t=2}^{T}%
p(y_{t}|y_{1:t-1},\theta), \label{likelihood}%
\end{equation}
where $y_{1:t-1}=\left(  y_{1},y_{2},\ldots,y_{t-1}\right)  ^{\prime},$ is
unavailable in closed form. Particle filtering algorithms play a role here by
producing (weighted)\emph{ }draws from the filtering density at time $t,$
$p(x_{t}|y_{1:t},\theta),$ with those draws in turn being used, via standard
calculations, to estimate the one-step ahead predictive densities of which the
likelihood function in (\ref{likelihood}) is comprised. The filtering
literature is characterized by different methods of producing and weighting
the filtered draws, or particles, with importance sampling principles being
invoked, and additional MCMC steps also playing a role in some cases. Not
surprisingly, performance of the alternative algorithms (often measured in
terms of the accuracy with which the filtered density itself is estimated) has
been shown to be strongly influenced by the empirical characteristics of the
SSM, with motivation for the development of a data-driven filter coming from
the poor performance of the BPF (in particular) in cases where the
signal-to-noise ratio is large; see\textbf{ }Giordani, Pitt and Kohn\textit{
}(2011) and Creal (2012) for extensive surveys and discussion, and Del Moral
and Murray (2015) for a more recent contribution.

A key insight of Andrieu\textit{ \textit{et al.} }(2010) is that particle
filtering can be used to produce an unbiased estimator of the likelihood
function which, when embedded within a suitable MCMC algorithm, yields exact
Bayesian inference, in the sense that the invariant distribution of the Markov
chain is the posterior of interest, $p(\theta|y_{1:T}).$ In brief, by defining
$u$ as the vector containing the canonical identically and independently
distributed ($i.i.d.$) random variables that underlie a given filtering
algorithm, and defining the corresponding\textbf{ }filtering-based estimate of
$L(\theta)$ by $\widehat{p}_{u}(y_{1:T}|\theta)=p(y_{1:T}|\theta,u)$, the role
played by the auxiliary variable $u$ in the production of the estimate is made
explicit. Andrieu\textit{ \textit{et al.}}\emph{ }demonstrate that under the
condition that%
\begin{equation}
E_{u}[\widehat{p}_{u}(y_{1:T}|\theta)]=p(y_{1:T}|\theta), \label{unbiased}%
\end{equation}
i.e., that\textbf{ }$\widehat{p}_{u}(y_{1:T}|\theta\mathbf{)}$ is an unbiased
estimator of the likelihood function,\textbf{ }then the marginal posterior
associated with the joint distribution,\emph{ }%
\begin{equation}
p(\theta,u|y_{1:T})\propto p(y_{1:T}|\theta,u)\times p(\theta)\times p\left(
u\right)  , \label{augmented joint}%
\end{equation}
is $p(\theta|y_{1:T}).$ Hence, this marginal posterior\textbf{ }density can be
accessed via an MH algorithm\emph{ }for example, in which the estimated
likelihood function, $\widehat{p}_{u}(y_{1:T}|\theta)$, replaces the exact
(but unavailable) $p(y_{1:T}|\theta)$ in the ratio that defines the MH
acceptance probability at iteration $i$ in the chain,%
\[
\alpha=\min\left\{  1,\frac{\widehat{p}\left(  y_{1:T}|\theta^{c}\right)
p\left(  \theta^{c}\right)  q\left(  \theta^{(i-1)}|\theta^{c}\right)
}{\widehat{p}\left(  y_{1:T}|\theta^{(i-1)}\right)  p\left(  \theta
^{(i-1)}\right)  q\left(  \theta^{c}|\theta^{(i-1)}\right)  }\right\}  ,
\]
where $q(\cdot)$ denotes the candidate density, $\theta^{c}$ a candidate draw
from $q(\cdot)$, and $\theta^{(i-1)}$ the value of the chain at iteration
$i-1$. Such an algorithm is referred to a particle \textit{marginal }MH (PMMH)
algorithm since draws from the augmented joint in (\ref{augmented joint})
represent draws from the desired marginal, $p(\theta|y_{1:T})$.

Flury and Shephard (2011) subsequently use this idea to conduct Bayesian
inference for a range of economic and financial models, employing the BPF as
the base particle filtering method. In addition, Pitt\textit{ \textit{et al.}
}(2012), drawing on Del Moral (2004), explicitly demonstrate the unbiased
property of the filtering-based likelihood estimators that are the focus of
their work and, as noted earlier, investigate the role played by the number of
particles in the resultant mixing of the chain.\emph{ }In summary, and as
might be anticipated, for any given particle filter\textbf{ }an increase in
the number of particles improves the precision of the corresponding\textbf{
}likelihood estimator (by decreasing its variance) and, hence, yields
efficiency that is arbitrarily close to that associated with an MCMC algorithm
that accesses the exact likelihood function. However, this accuracy\emph{ }is
typically\textbf{ }obtained at significant computational cost, with the
recommendation of Pitt\textit{ \textit{et al.}} being to choose the number of
particles that minimizes the cost of obtaining a precise likelihood estimator
yet still results in a sufficiently fast-mixing MCMC chain. Our aim is, in
part, to extend this analysis to cater for the DPF filters derived here and to
explore the performance of these new filters, both in a range of SSM settings
and in comparison with a number of\textbf{ }competing filters. However, the
overarching aim is to ascertain if differences in algorithmic performance -
arising from the use of different filters - translate into notable differences
in the PMMH-based forecast distributions and, hence, in forecast
accuracy.\footnote{Note, other PMCMC approaches, such as particle Gibbs, use
particle \emph{smoothing} techniques, in the production of $p\left(
\theta|y_{1:T}\right)  $, and bring with them additional numerical issues as a
consequence.(See Lindsten, Jordan and Sch\"{o}n, 2014, and Chopin and Singh,
2015.) Given our focus on forecasting \textit{per se}, as opposed to (joint)
parameter and state posterior inference, we focus solely on PMMH, and the
extraction from that algorithm of the filtered path of the states for use in
the forecasting exercise.}

\subsection{Forecasting with a PMMH algorithm\label{forecasting}}

Given the Markovian structure of (\ref{state}) and the condition provided by
(\ref{meas}), the one-step-ahead forecast density is given by%
\begin{equation}
p\left(  y_{T+1}|y_{1:T}\right)  =\int\int\int p\left(  y_{T+1}|x_{T+1}%
,\theta\right)  p\left(  x_{T+1}|x_{T},\theta\right)  p(x_{T}|y_{1:T}%
,\theta)p\left(  \theta|y_{1:T}\right)  dx_{T+1}dx_{T}d\theta.
\label{forecast}%
\end{equation}
To produce an estimate of this density, we start with a sequence of $MH$ draws
of $\theta$ drawn using the PMMH Markov chain, which result in the discrete
estimator of the posterior distribution $p\left(  \theta|y_{1:T}\right)  $,
given by%
\[
\widehat{p}\left(  \theta|y_{1:T}\right)  =\frac{1}{MH}\sum_{i=1}^{MH}%
\delta_{\theta^{(i)}},
\]
where $\delta_{\left(  \cdot\right)  }$ denotes the (Dirac) delta function,
see Au and Tam (1999).\footnote{Strictly speaking, $\delta_{\left(
\cdot\right)  }$ is a generalized function, and is properly defined as a
measure rather than as a function. However, we take advantage of the commonly
used heuristic definition here, as is also done in, for example, Ng.
\textit{et al} (2013).} Then, for each draw $\theta^{(i)}$, the one-step-ahead
predictive density, $p\left(  y_{T+1}^{o}|y_{1:T},\theta^{(i)}\right)  ,$ is
obtained through a final iteration of the particle filter, over a grid of
potential future observed values $\left\{  y_{T+1}^{o}\right\}  .$ The
estimate of (\ref{forecast}) is found by simply averaging these one-step-ahead
predictive densities, pointwise at each $y_{T+1}^{o}$ on the grid.

\section{The Filtering Algorithms\label{DPF_alg}}

\subsection{A quick review\label{overview}}

Particle filtering involves the sequential application of importance sampling
as each new observation becomes available, with the (incremental) target
density at\ time $t+1,$ being proportional to the\textbf{ }product of the
measurement density, $p\left(  y_{t+1}|x_{t+1},\theta\right)  $, and the
transition\textbf{ }density of the state $x_{t+1}$, denoted by $p\left(
x_{t+1}|x_{t},\theta\right)  $, as follows,
\begin{equation}
p(x_{t+1}|x_{t},y_{1:t+1},\theta)\propto p\left(  y_{t+1}|x_{t+1}%
,\theta\right)  p\left(  x_{t+1}|x_{t},\theta\right)  . \label{incr_target}%
\end{equation}
Note that draws of\textbf{ }the conditioning state value $x_{t}$\emph{ }are,
at time $t+1$, available from the previous iteration of the filter.

Particle\textbf{ }filters thus require the specification, at time $t+1$, of a
proposal density, denoted generically here by
\begin{equation}
g(x_{t+1}|x_{t},y_{1:t+1},\theta), \label{g}%
\end{equation}
from which the set of particles, $\{x_{t+1}^{[j]},j=1,...,N\}$, are generated
and, ultimately, used to estimate the filtered density as:%
\begin{equation}
\widehat{p}(x_{t+1}|y_{1:t+1},\theta)=\sum_{j=1}^{N}\pi_{t+1}^{[j]}%
\delta_{x_{t+1}^{[j]}}. \label{phat_xtplus1}%
\end{equation}
The normalized weights $\pi_{t+1}^{[j]}$ vary according to the choice of the
proposal $g\left(  \cdot\right)  $, the approach adopted for marginalization
(with respect to previous particles), and the way in which past particles are
`matched' with new particles in independent particle filter (IPF)-style
algorithms (Lin \textit{et al.}, 2005), of which the DPF is an example. In the
case of the BPF the proposal density in (\ref{g}) is equated to the transition
density, $p(x_{t+1}|x_{t},\theta)$, whilst for the APF the proposal is
explicitly dependent upon both the transition density and the observation
$y_{t+1}$, with the manner of the dependence determined by the exact form of
the auxiliary filter (see Pitt and Shephard, 1999, for details). The use of
both the observation and the previous state particle in the construction of a
proposal distribution is referred to as `adaptation' by the authors, with
`full' adaptation being feasible (only) when $p(y_{t+1}|x_{t},\theta)$ can be
computed and $p(x_{t+1}|x_{t},y_{t+1},\theta)$ is able to be\textbf{ }sampled
from directly. The UPF algorithm of van de Merwe\textit{ \textit{et al.}
}(2000) represents an alternative approach to adaptation, with particles
proposed\ via an approximation to the incremental target that uses unscented
transformations. For the IPF of Lin \textit{et al.} (2005), the proposal
reflects the form of $p\left(  y_{t+1}|x_{t+1},\theta\right)  $ in some
(unspecified) way, where the term `independence' derives from the lack of
dependence of the draws of $x_{t+1}$ on any previously obtained (and retained)
draws of $x_{t}.\emph{\ }$

With particle degeneracy (over time) being a well-known feature of filters, a
resampling step is typically employed. While most algorithms, including the
BPF, the UPF and the IPF, resample particles using the normalized weights
$\pi_{t+1}^{[j]}$, the APF incorporates resampling directly within $g\left(
\cdot\right)  $ by sampling particles from a\textbf{ }\textit{joint}
proposal,\textbf{ }$g(x_{t+1},k|x_{t},y_{1:t+1},\theta),$\textbf{ }where $k$
is an auxiliary variable that indexes previous particles. This allows the
resampling step, or the sampling of $k$, to take advantage of information from
the newly arrived observation,\textbf{ }$y_{t+1}$.

Given the product form of the target density $p(x_{t+1}|x_{t},y_{1:t+1}%
,\theta)$ in (\ref{incr_target}), the component that is relatively more
concentrated as a function of the argument\textbf{ }$x_{t+1}$\textbf{ }-
either\textbf{ }$p\left(  y_{t+1}|x_{t+1},\theta\right)  $\textbf{ }or\textbf{
}$p\left(  x_{t+1}|x_{t},\theta\right)  $ - will dominate in terms of
determining the shape of the target density. In the case of a strong
signal-to-noise ratio, meaning that the observation $y_{t+1}$ provides
significant information about the location of the unobserved state and
with\textbf{ }$p\left(  y_{t+1}|x_{t+1},\theta\right)  $\textbf{ }highly
peaked around $x_{t+1}$ as a consequence, an IPF proposal, which attempts to
mimic $p\left(  y_{t+1}|x_{t+1},\theta\right)  $ alone, can yield an accurate
estimate of\textbf{ }$p(x_{t+1}|x_{t},y_{1:t+1},\theta)$, in particular
out-performing the BPF, in which no account at all is taken of $y_{t+1}$ in
producing proposals of $x_{t+1}$.\emph{ }Lin \textit{et al.} (2005) in fact
demonstrate that, in a high signal-to-noise ratio scenario, an IPF-based
estimator of the mean of a filtered distribution can have a substantially
smaller variance than an estimator based on either the BPF or the APF,
particularly when computational time is taken into account. Since the DPF is
an IPF it\textbf{ }produces draws of $x_{t+1}$ via the structure of the
measurement density alone. The UDPF then augments the information from
$p\left(  y_{t+1}|x_{t+1},\theta\right)  $ with information from the second
component in (\ref{incr_target}). We detail these two new filters in the
following section.

\subsection{The new `data-driven' filters}

The key insight, first highlighted by Ng \textit{et al.} (2013) and motivating
the DPF and UDPF filters, is that the\textbf{ }measurement $y_{t+1}%
$\ corresponding to the state $x_{t+1}$\ in period $t+1$\ is often specified
via a measurement equation,%
\begin{equation}
y_{t+1}=h\left(  x_{t+1},\eta_{t+1}\right)  , \label{meas_fcn}%
\end{equation}
for a given function $h\left(  \cdot,\cdot\right)  $ and\textbf{ }$i.i.d.$
random variables $\eta_{t+1}$ having\textbf{ }common pdf $p\left(  \eta
_{t+1}\right)  $ , where $h\left(  \cdot,\cdot\right)  $ and $p\left(
\eta_{t+1}\right)  $ are $\theta$-dependent. Then, via a transformation of
variables, the measurement density may be expressed as
\begin{equation}
p(y_{t+1}|x_{t+1},\theta)=\int_{-\infty}^{\infty}p(\eta_{t+1})\left\vert
\frac{\partial h}{\partial x_{t+1}}\right\vert _{x_{t+1}=x_{t+1}(y_{t+1}%
,\eta_{t+1})}^{-1}\delta_{x_{t+1}(y_{t+1},\eta_{t+1})}d\eta_{t+1},
\label{Dirac_inversion}%
\end{equation}
where $x_{t+1}(y_{t+1},\eta_{t+1})$ is the unique solution to $y_{t+1}%
-h(x_{t+1},\eta_{t+1})=0$.\footnote{Extension to a finite number of multiple
roots is straightforward, and is discussed in Ng \textit{et al. }(2013).}
Further discussion of the properties of the representation in
(\ref{Dirac_inversion}) are provided in Ng\textit{ \textit{et al.}} The
advantage of the representation in (\ref{Dirac_inversion}) is that properties
of the delta function may be employed to manipulate the measurement density in
various ways. Whereas Ng \textit{et al.}\emph{ }exploit this representation
within a grid-based context, where the grid is imposed over the range of
possible values for the measurement error $\eta_{t+1}$, here we exploit the
representation to devise new particle filtering proposals, as detailed in the
following two subsections.

\subsubsection{The data-driven particle filter (DPF)\label{Section_DPF}}

With reference to (\ref{meas_fcn}), the DPF proposes particles by simulating
replicate and independent measurement errors, $\eta_{t+1}^{[j]}%
\overset{i.i.d.}{\sim}p(\eta_{t+1})$, and, given $y_{t+1}$, transforming these
draws to their implied state values $x_{t+1}^{[j]}=x_{t+1}(y_{t+1},\eta
_{t+1}^{[j]})$ via solution of the measurement equation. Recognizing the role
played by the Jacobian in (\ref{Dirac_inversion}), the particles
$x_{t+1}^{[j]}$ serve as a set of independent draws\emph{ }from a proposal
distribution with density $g(\cdot)$ satisfying
\begin{equation}
g(x_{t+1}|y_{t+1},\theta)=\left\vert \frac{\partial h\left(  x_{t+1}%
,\eta_{t+1}\right)  }{\partial x_{t+1}}\right\vert _{\eta_{t+1}=\eta^{\ast
}\left(  x_{t+1},y_{t+1}\right)  }p\left(  y_{t+1}|x_{t+1},\theta\right)  ,
\label{DPF_proposal_dist}%
\end{equation}
where $\eta^{\ast}\left(  x_{t+1},y_{t+1}\right)  $\emph{ }satisfies\emph{
}$y=h\left(  x_{t+1},\eta^{\ast}\left(  x_{t+1},y_{t+1}\right)  \right)
$.\emph{ }For the proposal distribution to have density $g(\cdot)$ in
(\ref{DPF_proposal_dist}), it is sufficient to assume both partial derivatives
of $h\left(  \cdot,\cdot\right)  $ exist and are non-zero, as is enforced in
the range of applications considered here. Furthermore, and given the lack of
explicit dependence of $g(\cdot)$ on $x_{t},$ the resultant sample from
(\ref{DPF_proposal_dist}) is such that the new draw $x_{t+1}^{[j]}$ can be
coupled with any previously simulated particle $x_{t}^{[i]}$, $i=1,2,...,N$.
When the $j^{th}$ particle $x_{t+1}^{[j]}$ is only ever matched with the
$j^{\emph{th}}$\ past\textbf{ }particle $x_{t}^{[j]}$, for each $j=1,...,N$
and each\textbf{ }$t=1,2,...,T$ , then the (unnormalized) weight\emph{ }of the
state draw is calculated as%
\begin{equation}
w_{_{t+1}}^{[j]}=\pi_{t}^{[j]}\frac{p\left(  y_{t+1}|x_{t+1}^{[j]}%
,\theta\right)  p\left(  x_{t+1}^{[j]}|x_{t}^{[j]},\theta\right)  }{g\left(
x_{t+1}^{[j]}|y_{t+1},\theta\right)  }, \label{IPF_MPF_weight}%
\end{equation}
for $j=1,...,N$. For the DPF, therefore, we have
\begin{equation}
w_{_{t+1}}^{[j]}=\pi_{t}^{[j]}\left\vert \frac{\partial h}{\partial x_{t+1}%
}\right\vert _{\eta_{t+1}=\eta_{t+1}^{[j]}\text{, }x_{t+1}=x_{t+1}^{[j]}}%
^{-1}p\left(  x_{t+1}^{[j]}|x_{t}^{[j]},\theta\right)  , \label{DPF_weight}%
\end{equation}
for $j=1,2,...,N$, where\textbf{ }$x_{0}^{[j]}\overset{iid}{\sim}p\left(
x_{0}|\theta\right)  $, $\pi_{0}^{[j]}=\frac{1}{N},$\textbf{ }and the
filtering weights\textbf{ }$\pi_{t+1}^{[j]}$,\textbf{ }in (\ref{phat_xtplus1}%
), are produced sequentially as
\begin{equation}
\pi_{t+1}^{[j]}\propto w_{_{t+1}}^{[j]} \label{particle_weights}%
\end{equation}
for all $j=1,2,...,N,$ with $\sum_{j=1}^{N}\pi_{t}^{[j]}=1$ for each $t$. In
addition, and as in any particle filtering setting (see, for example, Doucet
and Johansen, 2011), the iteration then provides component $t+1$ of the
estimated likelihood function as
\begin{equation}
\widehat{p}_{u}(y_{t+1}|y_{1:t},\theta)=\sum_{j=1}^{N}w_{t+1}^{[j]},
\label{PF_likelihood}%
\end{equation}
with each $w_{t+1}^{[j]}$ as given in (\ref{IPF_MPF_weight}).

Alternatively, as highlighted by Lin\textit{ \textit{et al.}} (2005), the
$j^{th}$ particle at $t+1$, could be matched with \textit{multiple} previous
particles from time\textbf{ }$t$. In this case, define $w_{t+1}^{[j][i]}$ as
the (unnormalized) weight corresponding to a match between $x_{t}^{[i]}$ and
$x_{t+1}^{[j]}$,
\[
w_{t+1}^{[j][i]}=\pi_{t}^{[i]}\frac{p\left(  y_{t+1}|x_{t+1}^{[j]}%
,\theta\right)  p\left(  x_{t+1}^{[j]}|x_{t}^{[i]},\theta\right)  }{g\left(
x_{t+1}^{[j]}|y_{1:t},\theta\right)  },
\]
for any $i=1,2,...,N$ and $j=1,2,...,N$. Next, denote $L$ \textit{distinct
cyclic} permutations\emph{ }of the elements in the sequence $(1,2,...,N)$
by\textbf{ }$K_{l}=(k_{l,1},...,k_{l,N})$, for $l=1,...,L.$ For each
permutation $l$, the $j^{th}$ particle $x_{t+1}^{[j]}$ is matched with the
relevant past particle indicated by\textbf{ }$x_{t}^{[k_{l,j}]}.$ Then, the
final (unnormalized) weight associated with $x_{t+1}^{[j]}$ is the simple
average, $w_{t+1}^{[j]}=\frac{1}{L}\sum_{l=1}^{L}w_{t+1}^{[j][k_{l,j}]}$.
Thus, in the DPF with multiple matching case, for $j=1,2,...,N$, we have
\begin{equation}
w_{_{t+1}}^{[j]}=\frac{1}{L}\left\vert \frac{\partial h}{\partial x_{t+1}%
}\right\vert _{\eta_{t+1}=\eta_{t+1}^{[j]}\text{, }x_{t+1}=x_{t+1}^{[j]}}%
^{-1}\sum_{l=1}^{L}\left[  \pi_{t}^{[k_{l,j}]}p\left(  x_{t+1}^{[j]}%
|x_{t}^{[k_{l,j}]},\theta\right)  \right]  , \label{Matched_DPF_weight}%
\end{equation}
with $\pi_{t}^{[j]}$ available from the previous iteration of the
filter.\textbf{ }Accordingly, as for the $L=1$ matching case in
(\ref{particle_weights}), the $\pi_{t+1}^{[j]}$ are then set proportional to
the $w_{_{t+1}}^{[j]}$ in (\ref{Matched_DPF_weight}), with\textbf{ }%
$\sum_{j=1}^{N}\pi_{t+1}^{[j]}=1$.\textbf{ }We consider this suggestion in
Section \ref{sim_ex}, and document the impact of the value of $L$ on filter
performance.\footnote{We note that the choice of $L=N$ yields the incremental
weight corresponding to the so-called marginal version of the filter. See
Klaas, de Freitas and Doucet (2012).} To ensure ease of implementation, pseudo
code for the DPF algorithm is provided as\textbf{ }Algorithm 1 in\textbf{
}Appendix A.1. Note that, as we implement the resampling of particles at each
iteration of all filters employed in Sections \ref{sim_ex} and
\ref{emp_section}, these steps are included as steps 8 and 9 in Algorithm 1.

The DPF, when applicable, thus provides a straightforward and essentially
automated way to estimate the likelihood function, in which only the
measurement equation is used in the generation of new particles.\footnote{This
idea of generating particles using only\textbf{ }information from the
observation and the measurement equation is, in fact, ostensibly similar to
notions of fiducial probability (see e.g. Hannig, Iyer, Lai and Lee, 2016).
However, in this case, whilst\textbf{ }the proposal density in
(\ref{DPF_proposal_dist}) for the latent state $x_{t+1}$ is obtained without
any knowledge of the predictive density (or `prior') given by $p\left(
x_{t+1}|y_{1:t},\theta\right)  $, the resampling of the proposed particles
according to either (\ref{DPF_weight}) or (\ref{Matched_DPF_weight}), means
that the resampled draws themselves \textit{do} take account of this
predictive information, and thus appropriately represent the filtered
distribution as in (\ref{phat_xtplus1}).} However, its performance depends
entirely on the extent to which the current observation is informative in
identifying the unobserved state location. This motivates the development of
the UDPF, in which proposed draws are informed by both the current observation
and a previous state particle.

\subsubsection{The unscented data-driven particle filter
(UDPF)\label{Section_UDPF}}

The UDPF uses a Gaussian approximation to the measurement density,\textbf{
}$p(y_{t+1}|x_{t+1},\theta)$,%
\begin{equation}
\widehat{p}(y_{t+1}|x_{t+1},\theta)\propto\frac{1}{\widehat{\sigma}_{M,t+1}%
}\phi\left(  \frac{x_{t+1}-\widehat{\mu}_{M,t+1}}{\widehat{\sigma}_{M,t+1}%
}\right)  . \label{DUPF_likelihood}%
\end{equation}
The terms $\widehat{\mu}_{M,t+1}$ and $\widehat{\sigma}_{M,t+1}^{2}$ denote
the (approximated)\ first and second (centred) moments (of $x_{t+1}$) implied
by an unscented transformation of\textbf{ }$\eta_{t+1}$\textbf{ }to\textbf{
}$x_{t+1}$,\textbf{ }for a given value of $y_{t+1}$, and where the subscript
$M$ is\textbf{ }used to reference the measurement equation via which these
moments are produced.\emph{ }Our motivation for using the unscented method,
including all details of the computation of the moments in this case, is
provided in Appendix A.2. The method derives from combining an assumed
Gaussian transition density, given by%
\begin{equation}
p(x_{t+1}|x_{t},\theta)=\frac{1}{\widehat{\sigma}_{P,t+1}}\phi\left(
\frac{x_{t+1}-\widehat{\mu}_{P,t+1}^{[j]}}{\widehat{\sigma}_{P,t+1}^{[j]}%
}\right)  , \label{state_tran}%
\end{equation}
where $\widehat{\mu}_{P,t+1}^{[j]}$ and $\widehat{\sigma}_{P,t+1}^{2[j]}$ are
assumed known, given particle\textbf{ }$x_{t}^{[j]}$, and the subscript $P$
references the predictive\textbf{ }transition equation to which these moments
apply. If needed, a Gaussian approximation (e.g. via a further unscented
transformation,\textbf{ }in which case both $\widehat{\mu}_{P,t+1}^{[j]}$ and
$\widehat{\sigma}_{P,t+1}^{2[j]}$ may also depend upon $x_{t}^{[j]}$) of a
non-Gaussian state equation may be accommodated. Having obtained the Gaussian
approximation in (\ref{DUPF_likelihood}) and applying the usual conjugacy
algebra, a proposal density is constructed as%
\begin{equation}
g(x_{t+1}|x_{t}^{[j]},y_{t+1},\theta)=\frac{1}{\widehat{\sigma}_{t+1}}%
\phi\left(  \frac{x_{t+1}-\widehat{\mu}_{t+1}^{[j]}}{\widehat{\sigma}%
_{t+1}^{[j]}}\right)  , \label{udpf_proposal}%
\end{equation}
where $\widehat{\mu}_{t+1}^{[j]}=\left\{  \widehat{\sigma}_{P,t+1}%
^{2[j]}\widehat{\mu}_{M,t+1}+\widehat{\sigma}_{M,t+1}^{2}\widehat{\mu}%
_{P,t+1}^{[j]}\right\}  /\left\{  \widehat{\sigma}_{P,t+1}^{2[j]}%
+\widehat{\sigma}_{M,t+1}^{2}\right\}  $ and $\widehat{\sigma}_{t+1}^{2[j]}=$
$\widehat{\sigma}_{M,t+1}^{2}\widehat{\sigma}_{P,t+1}^{2[j]}/$\newline%
$\left\{  \widehat{\sigma}_{P,t+1}^{2[j]}+\widehat{\sigma}_{M,t+1}%
^{2}\right\}  $ are the requisite moments of the Gaussian proposal, with the
bracketed superscript $[j]$ used\textbf{ }to reflect the dependence on the
$j^{th}$ particle\textbf{ }$x_{t}^{[j]}$. A resulting particle draw
$x_{t+1}^{[j]}$ from the proposal in (\ref{udpf_proposal}) is then weighted,
as usual, relative to the target, with the particle weight formula given by,%
\begin{equation}
w_{_{t+1}}^{[j]}=\pi_{t}^{[j]}\frac{p\left(  y_{t+1}|x_{t+1}^{[j]}%
,\theta\right)  p\left(  x_{t+1}^{[j]}|x_{t}^{[j]},\theta\right)  }{\frac
{1}{\widehat{\sigma}_{t+1}^{[j]}}\phi\left(  \frac{x_{t+1}^{[j]}-\widehat{\mu
}_{t+1}^{[j]}}{\widehat{\sigma}_{t+1}^{[j]}}\right)  }. \label{UDPF_weight}%
\end{equation}
The corresponding UDPF likelihood estimator is calculated as per
(\ref{PF_likelihood}). Pseudo code for the UDPF is provided in Algorithm 2 in
Appendix A.1, and proof of the unbiasedness of both data-driven filters is
given in Appendix A.3.

\section{Simulation Experiments\label{sim_ex}}

In this section, two different simulation exercises are undertaken. In the
first experiment we investigate, in a controlled setting, the performance of
the five filters, each described in Section \ref{DPF_alg}, in terms of their
ability to produce the posterior distribution for the model parameters. In
particular, in Section \ref{emp_app copy(1)}, we document both the
computational time and mixing performance of the relevant PMMH algorithm, for
three different state space models: the linear Gaussian (LG) model that is
foundational to all state space analysis, and two non-linear, non-Gaussian
models that feature in the empirical finance literature, namely the stochastic
conditional duration (SCD) and\textbf{ }SV specifications referenced in the
Introduction. Notably, different signal-to-noise ratios are entertained for
all three models. In Section \ref{forecast perf}, the forecast accuracy of the
estimate of (\ref{forecast}) yielded by each filtering method is considered
under, respectively, correct and incorrect specification of the true DGP.
Critically, having documented the differential performance of the alternative
filters in Section \ref{emp_app copy(1)}, we then deliberately hold both the
number of particles and the number of MCMC draws fixed in Section
\ref{forecast perf}, in order to guage the impact, or otherwise, of this
differential performance on forecast accuracy.

\subsection{Simulation design and PMMH evaluation methods\label{eval}}

Before detailing the specific design scenarios adopted for the simulation
exercises, we first define the signal-to-noise ratio (SNR)\ as\emph{ }%
\begin{equation}
SNR=\sigma_{x}^{2}/\sigma_{m}^{2}, \label{snr1}%
\end{equation}
where $\sigma_{x}^{2}$ is the unconditional variance of the state variable,
which is available analytically in all cases considered. In the LG setting
$\sigma_{m}^{2}$ corresponds directly to variance of the additive measurement
noise. In the two non-linear models,\textbf{ }a transformation of the
measurement equation is employed to enable the calculation of $\sigma_{m}^{2}%
$, now given by the variance of the (transformed) measurement error that
results, to be obtained either analytically (for the SV model) or using
deterministic integration (for the SCD model). The details of the relevant
transformations are provided in Sections \ref{sec_SCD} and \ref{sec_SV},
respectively.\textbf{ }The quantity in (\ref{snr1}) measures the strength of
signal relative to the background noise in the (appropriately transformed)
SSM, for given fixed parameter values.\footnote{A comparable quantity that is
applicable to the non-linear case is defined by $SNR^{\ast}=\sigma_{x}^{2}/V,$
where $V$ is the curvature of $\log(p(y_{t}|x_{t},\theta))$, see Giordani
\textit{et al}. (2011).}

A design scenario is defined by the combination of the model and corresponding
model parameter settings that achieve a given value for (\ref{snr1}), either
low or high. What constitutes a particular level for SNR is model-specific,
with values chosen (and reported below) that span the range of possible SNR
values that still accord with empirically plausible data. If a particular
design has a high SNR, this implies that observations are informative about
the location of the unobserved state. The DPF is expected to perform well in
this case, in terms of precisely estimating the (true) likelihood value, with
the impact of the use of multiple matching being of particular interest.
Conversely, as the BPF proposes particles from the state predictive
distribution, it is expected to have superior performance to the DPF when the
SNR is low. Exploiting both types of information at the same time, the UDPF,
APF and UPF methods are anticipated to be more robust to the SNR value.
However, when assessing PMMH performance and, ultimately forecast performance,
there is a more complex relationship between the filter performance and the
SNR of the DGP, given that estimation of the likelihood function takes place
across the full support of the unknown parameters.

The PMMH assessment draws on the insights of Pitt\textit{ \textit{et al.}}
(2012). If the likelihood is estimated precisely, the mixing of the Markov
chain will be as rapid as if the true likelihood were being used, and the
estimate of any posterior quantity will also be accurate as a consequence.
However, increasing the precision of the likelihood estimator by increasing
the number of particles used in the filter comes at a computational cost.
Equivalently, if a poor, but computationally cheap, likelihood estimator is
used within the PMMH algorithm, this will typically\textbf{ }slow the mixing
of the chain, meaning that for a given number of Markov chain iterates, the
estimate of any posterior quantity will be less accurate. Pitt\textit{
\textit{et al.}}\emph{ }suggest choosing the particle number that minimizes
the so-called computing time: a measure that takes into account both the cost
of obtaining the likelihood estimator and the speed of mixing of the MCMC
chain. They show that the `optimal' number of particles is that which yields a
variance for the likelihood estimator of 0.85, at the true parameter
vector.\footnote{Note, while the optimal number of particles may be computed
within a simulation context, as in the current section, implementation in an
empirical setting requires a preliminary estimate of the parameter (vector) at
which this computation occurs.}

For each design scenario (i.e. model and SNR level), and for each particular
filter, the PMMH algorithm is used to produce a Markov chain with $MH=110,000$
iterations, with the first $10,000$ iterations being discarded as burn-in. The
MCMC draws are generated from a random walk proposal, with the covariance
structure of the proposal\textbf{ }adapted using Algorithm 1 of M\"{u}ller
(2010).\emph{ }Determining the optimal number of particles, $N_{opt}$, for any
particular design scenario and for any specific filter, involves producing
$R_{0}$ independent replications of the likelihood estimate $\widehat{p}%
_{u}^{(r)}(y_{1:T}|\theta_{0}),$ each evaluated at the true parameter (vector)
$\theta_{0}$ and based on a selected number of particles, $N_{s}$. Then, the
optimum number of particles, denoted by $N_{opt}$, is chosen according
to\textbf{ }%
\begin{equation}
N_{opt}=N_{s}\times\frac{\widehat{\sigma}_{N_{s}}^{2}}{0.85}, \label{nopt}%
\end{equation}
where $\widehat{\sigma}_{N_{s}}^{2}$ denotes the variance of the $R_{0}$
likelihood estimates.\textbf{ }In other words, the (somewhat arbitrarily
selected) initial number of particles, $N_{s}$, is scaled\textbf{ }according
to the extent to which the precision that it is expected to yield (as
estimated by $\widehat{\sigma}_{N_{s}}^{2}$) varies from the value of $0.85$
that is sought. We track the time taken to compute the likelihood estimate at
each of the $MH$ iterations using the given filter with $N_{opt}$ particles.
We then record the average likelihood computing time (ALCT) over these
iterations, as well as the inefficiency factor (IF) for each parameter. In the
usual way, the IF for a given parameter can be interpreted as the sampling
variance of the mean of the correlated MCMC draws of that parameter relative
to the sampling variance of the mean of a hypothetical set of independent
draws. Values greater than unity thus measure the loss of precision (or
efficiency) incurred due to the dependence in the chain.

\subsection{Models, SNR ratios and priors}

In this section, we outline the three models used in the simulation
experiments. The parameter values and associated values for SNR are contained
in Table \ref{Simulation_parameter}.

\subsubsection{The linear Gaussian (LG) model}

The LG model is given by%
\begin{align}
y_{t}  &  =x_{t}+\sigma_{\eta}\eta_{t}\label{LL_meas}\\
x_{t}  &  =\rho x_{t-1}+\sigma_{v}v_{t}, \label{LL_state}%
\end{align}
with $\eta_{t}$ and $v_{t}$ mutually independent $i.i.d.$ standard normal
random variables. Data is generated using $\rho=0.4$ and $\sigma_{v}$ $=0.92$.
The value of $\sigma_{\eta}$ is set to achieve two values for SNR (low and
high), as recorded in Panel A of Table \ref{Simulation_parameter}. In the PMMH
exercises detailed in Section \ref{emp_app copy(1)}, where the parameters are
treated as unknown, the parameter $\theta=(\log(\sigma_{\eta}^{2}),\rho
,\log(\sigma_{v}^{2}))^{^{\prime}}$ is sampled\textbf{ }(thereby restricting
the simulated draws of $\sigma_{v}^{2}$ and $\sigma_{\eta}^{2}$ in the
resulting Markov chains to be positive), with a normal prior distribution
assumed as $\theta\sim N(\mu_{0},\Sigma_{0})$ with\textbf{ }$\mu_{0}%
=(\log(0.7),0.5,\log(0.475))^{\prime}$ and $\Sigma_{0}=I_{n}$. The same prior
is used in both high and low SNR settings, and is in the spirit of the prior
used in\textbf{ }Flury and Shephard (2011).%

\begin{table}[tbp] \centering
\caption{Parameters values used in the simulation exercises for the LG, SCD and SV models. The corresponding signal-to-noise ratio (SNR) for each scenario is shown in the bottom row.}\label{Simulation_parameter}%
\bigskip%

\begin{tabular}
[c]{cccccccccccc}\hline\hline
&  &  &  &  &  &  &  &  &  &  & \\
& \multicolumn{3}{c}{PANEL A: LG} &  & \multicolumn{3}{c}{PANEL B: SCD} &  &
\multicolumn{3}{c}{PANEL C: SV}\\\cline{1-4}\cline{5-12}
&  &  &  &  &  &  &  &  &  &  & \\
& Low &  & High &  & Low &  & High &  & Low &  & High\\\cline{2-4}%
\cline{2-4}\cline{6-8}\cline{10-12}
&  &  &  &  &  &  &  &  &  &  & \\
$\sigma_{\eta}$ & 2.24 &  & 0.45 & $\alpha$ & 0.67 &  & 6.67 & $\phi$ &
-6.61 &  & -4.24\\
$\rho$ & 0.40 &  & 0.40 & $\beta$ & 1.50 &  & 0.15 & $\rho$ & 0.2 &  & 0.6\\
$\sigma_{v}$ & 0.92 &  & 0.92 & $\phi$ & -1.1 &  & -1.1 & $\sigma_{v}$ &
0.70 &  & 1.40\\
&  &  &  & $\rho$ & 0.74 &  & 0.74 &  &  &  & \\
&  &  &  & $\sigma_{v}$ & 0.65 &  & 0.65 &  &  &  & \\
&  &  &  &  &  &  &  &  &  &  & \\
SNR & 0.2 &  & 5 &  & 0.5 &  & 10 &  & 0.1 &  & 0.6\\\hline\hline
\end{tabular}%
\end{table}%

\subsubsection{The stochastic conditional duration (SCD) model\label{sec_SCD}}

The SCD model is given by%
\begin{align}
y_{t}  &  =\exp(x_{t})\eta_{t}\label{scd_meas}\\
x_{t}  &  =\phi+\rho x_{t-1}+\sigma_{v}v_{t}, \label{scd_sate}%
\end{align}
with $v_{t}\sim i.i.d.N(0,1)$ independent of $\eta_{t},$ and with $\eta_{t}$
being\textbf{ }$i.i.d.$ from a gamma distribution\textbf{ }with shape
parameter $\alpha$ and rate parameter $\beta$. Taking the logarithms of both
sides of (\ref{scd_meas}) yields a transformed measurement equation that is
linear in the state variable\textbf{ }$x_{t}$, i.e. $\log(y_{t})=x_{t}%
+\varepsilon_{t},$ where $\varepsilon_{t}=\log(\eta_{t})$. The value
of\textbf{ }$\sigma_{m}^{2}=var(\varepsilon_{t})$ required to report the SNR
in Table \ref{Simulation_parameter} is obtained numerically. The initial state
is taken as the long run distribution of the state implied by choosing
$\left\vert \rho\right\vert <1$, that is\textbf{ }$x_{0}\sim N\left(
\frac{\phi}{1-\rho},\frac{\sigma_{v}^{2}}{(1-\rho)^{2}}\right)  $. For the
PMMH exercise detailed in Section \ref{emp_app copy(1)}, the parameter
vector\textbf{ }$\theta=(\log(\alpha),\log(\beta),\phi,\rho,\log(\sigma
_{v}^{2}))$\emph{ }is used to ensure the positivity of draws\textbf{ }for each
$\alpha$, $\beta$ and $\sigma_{v}^{2}$. As with the LG setting, a normal prior
is\textbf{ }adopted with $\theta\sim N(\mu_{0},\Sigma_{0})$, but now\textbf{
}with $\mu_{0}=(-0.8,0.5,\log(0.5),\log(2),\log(1))^{\prime}$ and $\Sigma
_{0}=I_{n}.$ This prior is again held constant over the two SNR settings (low
and high) used to assess PMMH performance.

\subsubsection{The stochastic volatility (SV) model\label{sec_SV}}

The SV\ model is given by%
\begin{align}
y_{t}  &  =\exp(x_{t}/2)\eta_{t}\label{SV_meas}\\
x_{t}  &  =\phi+\rho x_{t-1}+\sigma_{v}v_{t}, \label{SV_state}%
\end{align}
with $\eta_{t}$ and $v_{t}$ once again mutually independent sequences of
$i.i.d.$ standard normal random variables. To fix the SNR, the\textbf{
}measurement equation is transformed to\textbf{ }$\log(y_{t}^{2}%
)=x_{t}+\varepsilon_{t},$ where $\varepsilon_{t}=\log(\eta_{t}^{2})$. In this
case,\textbf{ }$var(\varepsilon_{t})=4.93$, corresponding to the quantity
$\sigma_{m}^{2}$ in (\ref{snr1}).\emph{ }The initial state distribution is
specified from the stationary distribution, with\textbf{ }$x_{0}\sim
N(\frac{\phi}{1-\rho},\frac{\sigma_{v}^{2}}{(1-\rho)^{2}})$. For the PMMH
exercise, a normal prior is adopted for $\theta=(\phi,\rho,\log(\sigma_{v}%
^{2}))^{^{\prime}}$, with $\theta\sim N(\mu_{0},\Sigma_{0})$, where $\mu
_{0}=(-4.6,0.8,\log(0.5))^{\prime}$ and $\Sigma_{0}=I_{n}$. This prior
is\textbf{ }used under both SNR settings.

\subsection{\textbf{Filter implementation details}}

The DPF and the UPDF are explained in detail in Sections\ \ref{Section_DPF}
and\ \ref{Section_UDPF}, respectively.\ The DPF is implemented with both $L=1$
and $L=30$ matches. Implementation of the BPF is standard, with details
available from many sources (e.g. Gordon\emph{ }\textit{\textit{et al.}},
1993, and Creal, 2012). The APF, on the other hand, may be implemented in a
variety of different ways, depending upon the model structure and the
preference of the analyst. For the models considered in this paper, so-called
full adaptation is feasible (only) for the LG model, and hence we report
results for this version of the filter (referred to as FAPF hereafter) in that
case. For all three models, we also experimented with an alternative version
of APF (in which full adaptation is not exploited) in which the proposal
distribution is given by $g\left(  x_{t+1},k|x_{t},y_{1:t+1},\theta\right)
=p(y_{t+1}|\mu(x_{t}^{[k]},\theta))p(x_{t+1}|x_{t}^{[k]},\theta),$ where
$\mu(x_{t}^{[k]})$ is the conditional mean $E(x_{t+1}|x_{t}^{[k]},\theta)$,
and $k$ is a discrete auxiliary variable (see Pitt and Shephard, 1999, for
details). For the SV model, we explored a third version of APF based on a
second order Taylor's series expansion of $\log(p(y_{t+1}|x_{t+1},\theta))$
around the maximum of the measurement density. This approach yields an
approximation of the likelihood component, denoted by $g(y_{t+1}%
|x_{t+1},\theta)$, which is then used to form a proposal distribution,
$g\left(  x_{t+1},k|x_{t},y_{1:t+1},\theta\right)  =g(y_{t+1}|x_{t+1}%
,\mu(x_{t}^{[k]},\theta))p(x_{t+1}|x_{t}^{[k]},\theta)$. (For more details,
see Pitt and Shephard, 1999, and Smith and Santos, 2006.) For the particular
non-linear models explored here, however, both non-fully-adapted APF methods
resulted in very unstable likelihood estimates. Hence, these filters were not
pursued further in either the documentation of PMMH performance results or the
production of PMMH-based forecast distributions.\footnote{Further details on
these results can be obtained from the authors on request.}

As is standard knowledge, the KF is a set of recursive equations suitable for
the LG model that enable calculation of the first two moments of the
distribution of the unobserved state variables given progressively observed
measurements.\emph{ }In a non-linear setting, the unscented Kalman filter uses
approximate Gaussian distributions obtained from the unscented transformations
applied within the recursive KF structure, to approximate each of the
(non-Gaussian) filtered state distributions. In contrast, the UPF that is
implemented in our setting, uses approximate Gaussian distributions for the
proposal distributions in (\ref{g}) with moments produced by the unscented
transformations, and\textbf{ }with the conditioning on each new observation
$y_{t+1}$ obtained\textbf{ }as if the model were an LG model with moments that
match those of the conditional distributions defined by $p\left(
y_{t+1}|x_{t+1},\theta\right)  $ and $p\left(  x_{t+1}|x_{t},\theta\right)
$.\textbf{ }Further discussion of the UPF is provided in van de Merwe\textit{
\textit{et al.}} (2000).

\subsection{PMMH performance: Simulation results\label{emp_app copy(1)}}

At each MCMC\ iteration, the particle filter based on $N_{opt}$ is used to
estimate the likelihood function conditional on the set of parameter values
drawn at that iteration. The value of $N_{opt}$, however, is determined (via
the preliminary exercise described in Section \ref{eval}, with $R_{0}=100$
replications and $N_{s}=1000$ particles) at the true parameter values only
and, hence, is influenced by the SNR associated with the true data generating
process. Thus, when considering the performance of the filters within an MCMC
algorithm two things are required: 1) efficient performance at the SNR for the
true data, leading to a small value of $N_{opt}$; plus 2) some robustness in
performance to the SNR, since the movement across the parameter space (within
the chain) effectively changes the SNR under which the likelihood function is
computed at each point. A small value of $N_{opt}$ will, \textit{ceteris
paribus}, tend to produce a small value for the ALCT and, thus, ease the
computational burden. However, a lack of robustness of the filter will lead to
inaccurate likelihood estimates and, hence poor mixing in the chain.\textbf{
}Both the ALCT and the IF thus need to be reported for each filter, and for
each model, with the preferable filter being that which yields acceptable
mixing performance in a\textbf{ }reasonable time across for all three models.
The results documented in this section are based on a sample size of $T=250$,
reflecting the need for at least a moderate sample size when comparing the
performance of competing \textit{inferential} algorithms in a state space setting.%

\begin{table}[tbp] \centering
\caption{LG model: The optimal number of particles, average likelihood computing time (ALCT) and the inefficiency factor (IF) are reported for the PMCMC algorithm using DPF (with $L=1$ and 30 matches), BPF, UPF, UDPF and FAPF to produce the likelihood estimator. Data is simulated from the model in (22) and (23) with SNR = 0.2 in the top panel and SNR = 5 in the bottom panel.}\label{PMCMC_result}%

\begin{tabular}
[c]{llllllll}\hline\hline
&  &  &  &  &  &  & \\
\multicolumn{8}{c}{PMMH results under a low SNR setting}\\
&  &  &  &  &  &  & \\
&  & $N_{opt}$ & ALCT &  & \multicolumn{3}{c}{IF}\\\cline{3-4}\cline{6-8}
&  &  &  &  & \multicolumn{1}{c}{$\sigma_{\eta}^{2}$} &
\multicolumn{1}{c}{$\rho$} & \multicolumn{1}{c}{$\sigma_{v}^{2}$}\\
DPF ($L=1$) &  & 379 & 0.100 &  & \multicolumn{1}{r}{241.4} &
\multicolumn{1}{r}{325.2} & \multicolumn{1}{r}{248.6}\\
DPF ($L=30$) &  & 348 & 0.647 &  & \multicolumn{1}{r}{300.2} &
\multicolumn{1}{r}{321.3} & \multicolumn{1}{r}{313.1}\\
BPF &  & 18 & 0.017 &  & \multicolumn{1}{r}{184.5} & \multicolumn{1}{r}{67.7}
& \multicolumn{1}{r}{178.6}\\
UPF &  & 57 & 0.057 &  & \multicolumn{1}{r}{128.2} & \multicolumn{1}{r}{88.3}
& \multicolumn{1}{r}{144.5}\\
UDPF &  & 4 & 0.065 &  & \multicolumn{1}{r}{134.2} & \multicolumn{1}{r}{63.3}
& \multicolumn{1}{r}{149.9}\\
FAPF &  & 2 & 0.022 &  & \multicolumn{1}{r}{163.8} & \multicolumn{1}{r}{98.4}
& \multicolumn{1}{r}{181.6}\\\hline
&  &  &  &  &  &  & \\
\multicolumn{8}{c}{PMMH results under a\ high SNR setting}\\
&  & $N_{opt}$ & ALCT &  & \multicolumn{3}{c}{IF}\\\cline{3-4}\cline{6-8}
&  &  &  &  & \multicolumn{1}{c}{$\sigma_{\eta}^{2}$} &
\multicolumn{1}{c}{$\rho$} & \multicolumn{1}{c}{$\sigma_{v}^{2}$}\\
DPF ($L=1$) &  & 168 & 0.084 &  & \multicolumn{1}{r}{33.3} &
\multicolumn{1}{r}{31.9} & \multicolumn{1}{r}{33.3}\\
DPF ($L=30$) &  & 143 & 0.323 &  & \multicolumn{1}{r}{33.1} &
\multicolumn{1}{r}{33.6} & \multicolumn{1}{r}{35.9}\\
BPF &  & 2750 & 0.254 &  & \multicolumn{1}{r}{20.9} & \multicolumn{1}{r}{20.0}
& \multicolumn{1}{r}{20.4}\\
UPF &  & 70 & 0.066 &  & \multicolumn{1}{r}{31.6} & \multicolumn{1}{r}{31.6} &
\multicolumn{1}{r}{32.2}\\
UDPF &  & 23 & 0.060 &  & \multicolumn{1}{r}{37.4} & \multicolumn{1}{r}{35.5}
& \multicolumn{1}{r}{39.9}\\
FAPF &  & 11 & 0.035 &  & \multicolumn{1}{r}{35.0} & \multicolumn{1}{r}{35.4}
& \multicolumn{1}{r}{39.8}\\\hline\hline
\end{tabular}%
\end{table}%
%

\begin{table}[tbp] \centering
\caption{SCD model: The optimal number of particles, average likelihood computing time (ALCT) and the inefficiency factor (IF) are reported for the PMCMC algorithm using DPF (with $L=1$ and 30 matches), BPF, UPF and UDPF to produce the likelihood estimator. Data is simulated from the model in (24) and (25) with SNR = 0.5 in the top panel and SNR = 10 in the bottom panel.}\label{PMCMC_result_SCD}%

\begin{tabular}
[c]{llllllllll}\hline\hline
&  &  &  &  &  &  &  &  & \\
\multicolumn{10}{c}{PMMH results under a low SNR setting}\\
&  &  &  &  & \multicolumn{1}{r}{} & \multicolumn{1}{r}{} &
\multicolumn{1}{r}{} & \multicolumn{1}{r}{} & \multicolumn{1}{r}{}\\
&  & $N_{opt}$ & ALCT &  & \multicolumn{5}{c}{IF}\\\cline{3-4}\cline{6-10}
&  &  &  &  & \multicolumn{1}{c}{$\phi$} & \multicolumn{1}{c}{$\rho$} &
\multicolumn{1}{c}{$\sigma_{v}^{2}$} & \multicolumn{1}{c}{$\alpha$} &
\multicolumn{1}{c}{$\beta$}\\
DPF ($L=1$) &  & 1469 & 0.346 &  & 50.3 & 61.4 & 34.7 & 39.3 & 52.8\\
DPF ($L=30$) &  & 1296 & 2.1 &  & 47.7 & 45.6 & 38.8 & 37.2 & 45.9\\
BPF &  & 184 & 0.064 &  & 58.3 & 57.3 & 46.1 & 40.0 & 42.6\\
UPF &  & 398 & 0.239 &  & 55.2 & 67.1 & 49.6 & 42.3 & 66.3\\
UDPF &  & 119 & 0.065 &  & 45.9 & 55.9 & 52.3 & 36.8 & 55.1\\\hline
&  &  &  &  &  &  &  &  & \\
\multicolumn{10}{c}{PMMH results under a high SNR setting}\\
&  &  &  &  & \multicolumn{1}{r}{} & \multicolumn{1}{r}{} &
\multicolumn{1}{r}{} & \multicolumn{1}{r}{} & \multicolumn{1}{r}{}\\
&  & $N_{opt}$ & ALCT &  & \multicolumn{5}{c}{IF}\\\cline{3-4}\cline{6-10}
&  &  &  &  & \multicolumn{1}{c}{$\phi$} & \multicolumn{1}{c}{$\rho$} &
\multicolumn{1}{c}{$\sigma_{v}^{2}$} & \multicolumn{1}{c}{$\alpha$} &
\multicolumn{1}{c}{$\beta$}\\
DPF ($L=1$) &  & 177 & 0.060 &  & \multicolumn{1}{r}{55.6} &
\multicolumn{1}{r}{51.7} & \multicolumn{1}{r}{48.0} & \multicolumn{1}{r}{49.2}
& \multicolumn{1}{r}{44.4}\\
DPF ($L=30$) &  & 159 & 0.381 &  & \multicolumn{1}{r}{40.8} &
\multicolumn{1}{r}{43.9} & \multicolumn{1}{r}{47.2} & \multicolumn{1}{r}{48.9}
& \multicolumn{1}{r}{36.0}\\
BPF &  & 1011 & 0.218 &  & \multicolumn{1}{r}{45.8} & \multicolumn{1}{r}{44.6}
& \multicolumn{1}{r}{46.1} & \multicolumn{1}{r}{62.0} &
\multicolumn{1}{r}{34.5}\\
UDPF &  & 73 & 0.068 &  & 62.4 & 67.6 & 69.5 & 72.6 & 46.8\\
UPF &  & 118 & 0.113 &  & \multicolumn{1}{r}{65.1} & \multicolumn{1}{r}{58.7}
& \multicolumn{1}{r}{54.7} & \multicolumn{1}{r}{54.8} &
\multicolumn{1}{r}{54.6}\\\hline\hline
\end{tabular}%
\end{table}%
%

\begin{table}[tbp] \centering
\caption{SV model: The optimal number of particles, average likelihood computing time (ALCT) and the inefficiency factor (IF) are reported for the PMCMC algorithm using DPF (with $L=1$ and 30 matches), BPF, UPF and UDPF to produce the likelihood estimator. Data is simulated from the model in (26) and (27) with SNR = 0.1 in the top panel and SNR = 0.6 in the bottom panel.}\label{PMCMC_result_SV}%

\begin{tabular}
[c]{llllllrrr}\hline\hline
&  &  &  &  &  &  &  & \\
\multicolumn{9}{c}{PMMH results under a low SNR setting}\\
&  &  &  &  &  &  &  & \\
&  & $N_{opt}$ & ALCT &  &  & \multicolumn{3}{c}{IF}\\\cline{3-4}\cline{7-9}
&  &  &  &  &  & \multicolumn{1}{c}{$\phi$} & \multicolumn{1}{c}{$\rho$} &
\multicolumn{1}{c}{$\sigma_{v}^{2}$}\\
DPF (L=1) &  & 1767 & 0.534 &  &  & 22.53 & 22.42 & 19.07\\
DPF (L=30) &  & 1548 & 3.37 &  &  & 21.39 & 21.38 & 18.59\\
BPF &  & 275 & 0.067 &  &  & 14.82 & 14.83 & 16.17\\
UPF &  & 244 & 0.133 &  &  & 22.35 & 22.32 & 21.35\\
UDPF &  & 245 & 0.085 &  &  & \multicolumn{1}{l}{14.55} &
\multicolumn{1}{l}{14.62} & \multicolumn{1}{l}{13.97}\\\hline
&  &  &  &  &  & \multicolumn{1}{l}{} & \multicolumn{1}{l}{} &
\multicolumn{1}{l}{}\\
\multicolumn{9}{c}{PMMH results under a\ high\textbf{ }SNR setting}\\
&  &  &  &  &  & \multicolumn{1}{l}{} & \multicolumn{1}{l}{} &
\multicolumn{1}{l}{}\\
&  & $N_{opt}$ & ALCT &  &  & \multicolumn{3}{c}{IF}\\\cline{3-4}\cline{7-9}
&  &  &  &  &  & \multicolumn{1}{c}{$\phi$} & \multicolumn{1}{c}{$\rho$} &
\multicolumn{1}{c}{$\sigma_{v}^{2}$}\\
DPF ($L=1$) &  & 1013 & 0.265 &  &  & 14.85 & 15.39 & 15.45\\
DPF ($L=30$) &  & 915 & 1.889 &  &  & 17.59 & 18.25 & 16.39\\
BPF &  & 605 & 0.104 &  &  & 16.86 & 16.62 & 14.94\\
UPF &  & 686 & 0.256 &  &  & 16.41 & 16.53 & 18.84\\
UDPF &  & 568 & 0.156 &  &  & \multicolumn{1}{l}{12.55} &
\multicolumn{1}{l}{12.91} & \multicolumn{1}{l}{13.81}\\\hline\hline
\end{tabular}%
\end{table}%

The PMMH results for the LG model are presented in Table \ref{PMCMC_result}.
As is consistent with expectations, under the high SNR setting, the optimum
number of particles for the BPF is much larger than that for the DPF. This
then translates into higher values for ALCT for the BPF than for the DPF, when
a single match only ($L=1$) is used. Further reduction in $N_{opt}$ is yielded
via the multiple matching ($L=30$), via the extra precision that is produced
from the averaging process. However, this comes at a distinct cost in
computational time, with the gain of the DPF over the BPF, in terms of ALCT,
lost as a consequence. In the low SNR setting, also as anticipated, the basic
DPF (for either value of $L$) does not produce gains over the BPF, either in
terms of $N_{opt}$ or ALCT.

In contrast to the variation in the performance of the DPF - relative to the
BPF - over the SNR settings, the UDPF is uniformly superior to the BPF in
terms of $N_{opt}$, with the increase in computational cost associated with
the likelihood estimation (as a consequence of having to perform the unscented
transformations) resulting in only a slightly larger value for ALCT (relative
to that for the BPF) in the low SNR case. Moreover, the UDPF yields very
similar values of $N_{opt}$ to the analytically\textbf{ }available FAPF and
values for ALCT that are not much higher. The values of $N_{opt}$ for the UDPF
are also much lower than those for the UPF, with ALCT being only slightly
larger for the former in the low SNR case.

As one would anticipate, given that $N_{opt}$ for each filter is deliberately
selected to ensure a given level of accuracy in the estimation of the
likelihood (albeit at the true parameter values only), the variation in the
IFs (for any given parameter) across the different filters is not particularly
marked. That said, there are still some differences, with the UDPF, along with
the UPF, being the best performing filters overall, when both SNR scenarios in
this LG setting are considered, and the DPF (for both values of $L$) being the
most inefficient filter in the low SNR case.

The PMMH results for the SCD and SV models are presented in Table
\ref{PMCMC_result_SCD} and \ref{PMCMC_result_SV} respectively. Both sets of
results are broadly similar to those for the LG model in terms of the relative
performance of the methods, remembering that the FAPF is not applicable in the
non-linear case and all other versions of the APF are eschewed due to the poor
likelihood estimation results cited earlier. For the SCD model, the
conclusions drawn above regarding the relative performance of the BPF and DPF
filters apply here also. In this case, however, when all three factors:
robustness to SNR, ALCT value and IF value are taken into account, the UDPF is
uniformly superior to all other filters. For the SV model, as the `high' SNR
value appears relatively small, set as such to ensure that the model produces
empirically plausible data, the DPF has less of a comparative advantage over
the BPF. However, the UDPF is competitive with the (best performing) BPF in
both settings, according to ALCT, and is uniformly superior to all other
filters according to the IF values.

Overall then, when robustness to SNR, computation time and chain performance
are all taken into account the UDPF is the preferred choice for the
experimental designs considered here. We now address the question of\textbf{
}what difference, if any, this superiority in (algorithmic) performance makes
at the forecasting level.

\subsection{Forecast performance: Simulation results\label{forecast perf}}

The impact of the particle filter on forecasting is explored in the context of
estimating the SV model described previously in Section \ref{sec_SV}. We
\textit{simulate} the data, however, under two different scenarios - one where
the SV model is correctly specified, and the other where the SV model does not
correspond to the true DGP. In the first case, data are simulated under the SV
model in (\ref{SV_meas}) and (\ref{SV_state}), using the low SNR setting shown
in Panel C of Table \ref{Simulation_parameter}. In this second case, the data
follows (a discrete approximation to) a bivariate jump diffusion process, with
independent random jumps sporadically occurring, in the price and/or
volatility process (see Duffie, Pan and Singleton, 2000). In this context, a
price jump relates to large observed deviation from the expected return,
whereas a volatility jump corresponds to an unusually large deviation in the
underlying volatility process. We refer to this DGP as the stochastic
volatility with independent jumps (SVIJ) model.

According to the SVIJ model the observed value $y_{t}$ is generated as
\begin{align}
y_{t}  &  =\sqrt{x_{t}}\zeta_{t}^{p}+Z_{t}^{p}\Delta N_{t}^{p}\label{svij1}\\
x_{t}  &  =\kappa\theta+(1-\kappa)x_{t-1}+\sigma_{v}\sqrt{x_{t}}\zeta_{t}%
^{x}+Z_{t}^{x}\Delta N_{t}^{x}, \label{svij2}%
\end{align}
with $\zeta_{t}^{p}$ and $\zeta_{t}^{x}$ being independent sequences of
$i.i.d.$ standard normal random variables. The jump components of the
measurement equation and the state equation are both composed from two
separate parts: jump occurrence and size. The jump occurrence sequences have
elements denoted by $\Delta N_{t}^{p}$ and $\Delta N_{t}^{x},$ respectively,
and are each independent $i.i.d.$ Bernoulli random variables taking the value
one with $15\%$ and $20\%$ probability, respectively. The size of each jump in
the state equation, $Z_{t}^{x},$ is generated from an exponential
distribution, with a mean value of 0.02. The size of a price jump, $Z_{t}^{p}%
$, is generated as $Z_{t}^{p}=S_{t}\exp\left(  M_{t}\right)  $, where $S_{t}$
is equal to either $+1$ or $-1$ with equal probability, and where $M_{t}$ is
$i.i.d.$ from a normal distribution with mean zero and variance equal to 0.5.
These numerical values were selected to accord with estimated values obtained
from the empirical study of the S\&P 500 index undertaken in Maneesoonthorn,
Forbes and Martin (2017).

In both the correctly and incorrectly specified cases, a single time series of
length $T=750$ is generated, with the first $500$ observations used to produce
competing PMMH-based estimates of the posterior distribution $p\left(
\theta|y_{1:500}\right)  $, where $\theta=(\phi,\rho,\sigma_{v}^{2})^{\prime}%
$, i.e. the parameters of the estimated SV model. At each PMMH iteration, a
candidate draw of the transformed parameter, $\widetilde{\theta}=(\phi
,\rho,\log(\sigma_{v}^{2}))^{\prime}$, is generated from a random walk
proposal. The prior distribution of $\widetilde{\theta}$ is specified as:
$\phi\sim N(0,10),$ $\rho\sim Beta(20,1.5),$ $\log(\sigma_{v}^{2})\sim
N(0,10).$\textbf{ }For each of the four remaining\textbf{ }alternative
filters, DPF (with $L=1$), UDPF, BPF and UPF, $N=300$ particles are used to
estimate the likelihood function at each PMMH iteration, and $MH=5000$
iterations drawn. By holding the number of particles used in each filter fixed
at a common value, the resulting efficiency associated with the estimated
likelihood function, and hence the PMMH algorithm itself (for a given number
of MCMC draws), will be different. The impact of controlling the number of
particles on the IFs of the resultant Markov chains is evident in Table
\ref{IF}, where it is clear that some Markov chains are more efficient than
others. Notably, the DPF is relatively inefficient compared to the other three
filters and, overall, the UDPF continues to exhibit the superior performance
documented in the previous section.%

\begin{table}[tbp] \centering
\caption{The inefficiency factors obtained from the PMMH algorithms, with the DPF, BPF, UPF and UDPF used to produce the likelihood estimator. The data is simulated from the SV model (columns 2-4) and the SVIJ model (columns 5-7).}\label{IF}%
\begin{tabular}
[c]{lllllllll}\hline\hline
&  &  &  &  &  &  &  & \\
\multicolumn{9}{c}{Inefficiency factors}\\
&  &  &  &  &  &  &  & \\
&  & \multicolumn{3}{c}{SV} &  & \multicolumn{3}{c}{SVIJ}\\\cline{3-8}%
\cline{3-9}
&  & \multicolumn{1}{c}{$\phi$} & \multicolumn{1}{c}{$\rho$} &
\multicolumn{1}{c}{$\sigma_{v}^{2}$} &  & \multicolumn{1}{c}{$\phi$} &
\multicolumn{1}{c}{$\rho$} & \multicolumn{1}{c}{$\sigma_{v}^{2}$}\\
DPF &  & 386.51 & 356.58 & 370.14 &  & 354.02 & 353.60 & 301.93\\
BPF &  & 23.59 & 23.56 & 20.55 &  & 20.81 & 20.77 & 18.67\\
UPF &  & 80.19 & 80.46 & 72.49 &  & 40.72 & 10.78 & 32.90\\
UDPF &  & 22.95 & 23.01 & 22.13 &  & 18.79 & 18.82 & 16.06\\\hline\hline
\end{tabular}
\end{table}%

We then proceed to estimate the competing\ (marginal) one-step-ahead forecast
distributions, for each of $250$ subsequent periods, each time following the
procedure described in Section \ref{forecasting}.\footnote{Due to the
intensive nature of the PMMH algorithm, the posterior draws of $\theta$ are
refreshed only after 50 forecast periods.} Due to the singularity of the
square root at zero, we produce forecast distributions for the transformed
measurement $\log\left(  y_{T+k}^{2}\right)  $. Having produced these
competing forecast distributions, their performance is measured using the
average log score, which we denote by ALS and calculate as the average of the
logarithm of each estimated predictive density $\widehat{p}\left(  \log\left(
y_{T+k}^{2}\right)  |y_{1:T+k-1}\right)  $, evaluated at the subsequently
`realised' value, $\log\left(  \left(  y_{T+k}^{obs}\right)  ^{2}\right)  $,
for $k=1,2,...,250$. We also compute the average absolute difference between
the log score produced under the BPF and that produced by each of the other
three filters. We denote this average absolute difference in log scores by ADLS.

Despite the four filters having quite different inefficiency factors, we find
that the forecast accuracy is virtually unaffected by which filter is used,
and irrespective of whether the fitted model is correctly or incorrectly
specified. Figures \ref{forecast_dist_w_fix} and \ref{forecast_dist_jump_fix}
illustrate the estimated forecast forecast distributions of the first
out-of-sample period (i.e. with $T=500$ and $k=1$), when the data are
generated by the SV DGP and the SVIJ DGP, respectively. The top panels in
Figures \ref{forecast_dist_w_fix} and \ref{forecast_dist_jump_fix} display all
of the estimated conditional one-step-ahead forecast distributions associated
with each of the filters, i.e. all of the $\widehat{p}\left(  \log\left(
y_{T+k}^{2}\right)  |y_{1:T+k-1},\theta^{(i)}\right)  $ for $i=1,2,...,5000$
MH iterations and for each filter. While all of the conditional forecast
distributions produced by each of the four filters appear to be centered
around a similar location, the DPF and UPF produce more varied conditional
forecasts than do either the BPF or UDPF. However, since the marginal
one-step-ahead forecast distribution is produced by integrating out the
uncertainty associated with the unknown parameters, much of the variation
between the conditional forecast distributions is eliminated through the
averaging procedure. Hence, the competing\ estimated marginal one-step-ahead
forecast distributions, shown in the bottom panels of Figures
\ref{forecast_dist_w_fix} and \ref{forecast_dist_jump_fix} for the correctly
and incorrectly specified SV models, respectively, are visually
indistinguishable from each other. In addition, as shown in Table
\ref{average_ls}, we find no difference (to two decimal places) in the ALS
produced from the $250$ one-step-ahead forecasts, with similarly\textbf{
}negligible results obtained for the ADLS.%

\begin{table}[tbp] \centering
\caption{ALS and ADLS, obtained from estimating the SV model using PMMH with four different filters.
The simulated data is simulated from the the SV (columns 2 and 3) and SVIJ (columns 4 and 5) specification, respectively.}\label{average_ls}%
\bigskip%
\begin{tabular}
[c]{llllllll}\hline\hline
&  &  &  &  &  &  & \\
\multicolumn{8}{c}{Log score summaries}\\
&  & \multicolumn{5}{c}{} & \\
DGP &  & \multicolumn{2}{c}{SV} &  &  & \multicolumn{2}{l}{SVIJ}%
\\\cline{3-6}\cline{3-8}
&  &  &  &  &  &  & \\
&  & ALS & ADLS &  &  & ALS & ADLS\\
DPF &  & -2.2053 & 0.0240 &  &  & -2.1862 & 0.0298\\
BPF &  & -2.2071 & 0 &  &  & -2.1854 & 0\\
UPF &  & -2.2070 & 0.0047 &  &  & -2.1852 & 0.00352\\
UDPF &  & -2.2072 & 0.0032 &  &  & -2.1856 & 0.0035\\\hline\hline
\end{tabular}%
\end{table}%

\bigskip\begin{figure}[ptb]
\begin{center}
\includegraphics[angle=90,origin=c,scale=0.6]
{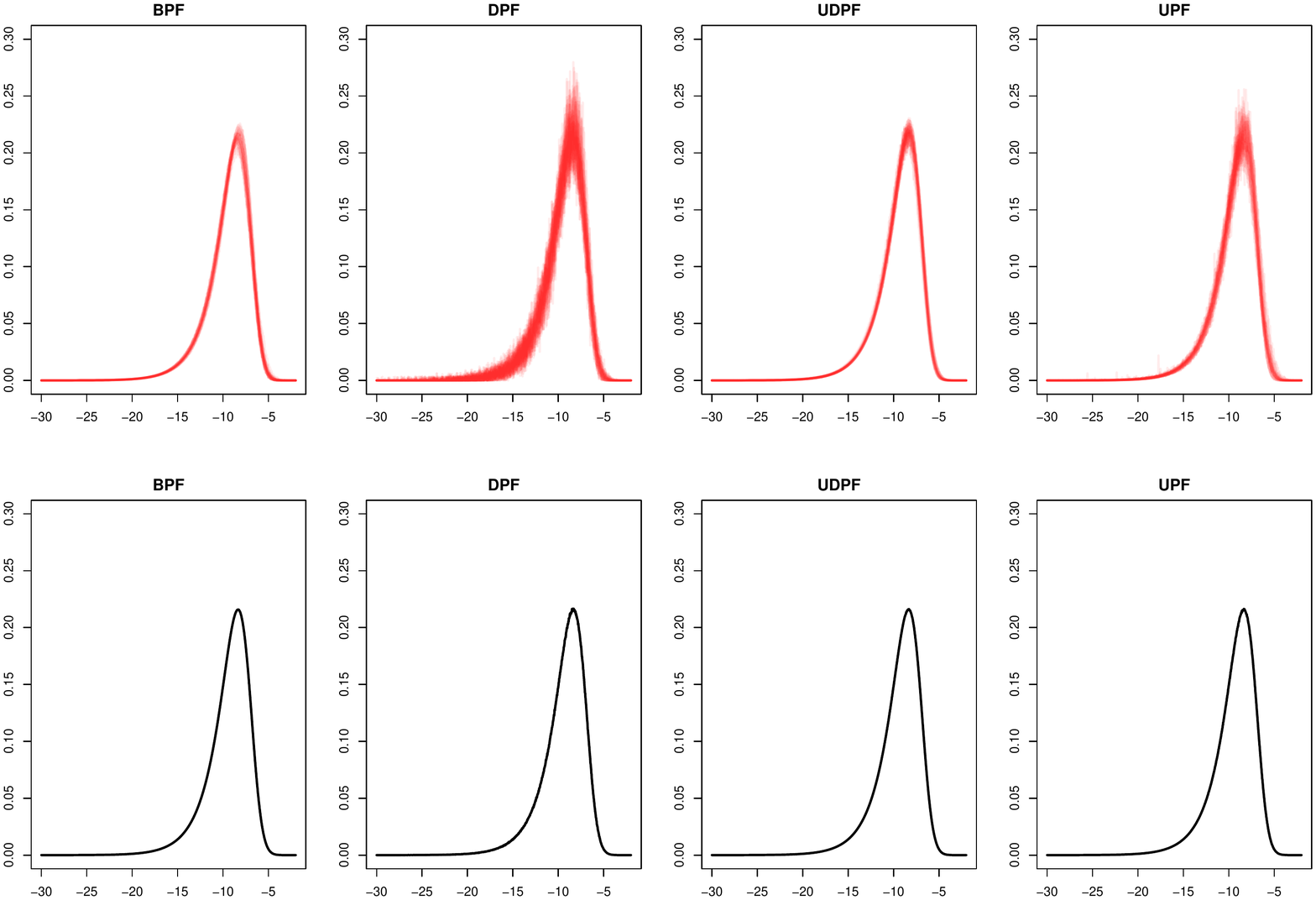}
\end{center}
\caption{The top panel shows all of the individual \textit{conditional
}one-step ahead forecast distributions of $\log(y_{501}^{2})$ produced by BPF
(in black), UPF (in green), DPF (in red) and UDPF (in blue), for the correctly
specified SV model, using simulated data. The bottom panel displays, for each
filter, the \textit{marginal} one-step ahead forecast distribution, which is
obtained using the average of the corresponding conditional forecasts shown in
the top panel. }%
\label{forecast_dist_w_fix}%
\end{figure}

\begin{figure}[ptb]
\begin{center}
\includegraphics[angle=90,origin=c,scale=0.6]
{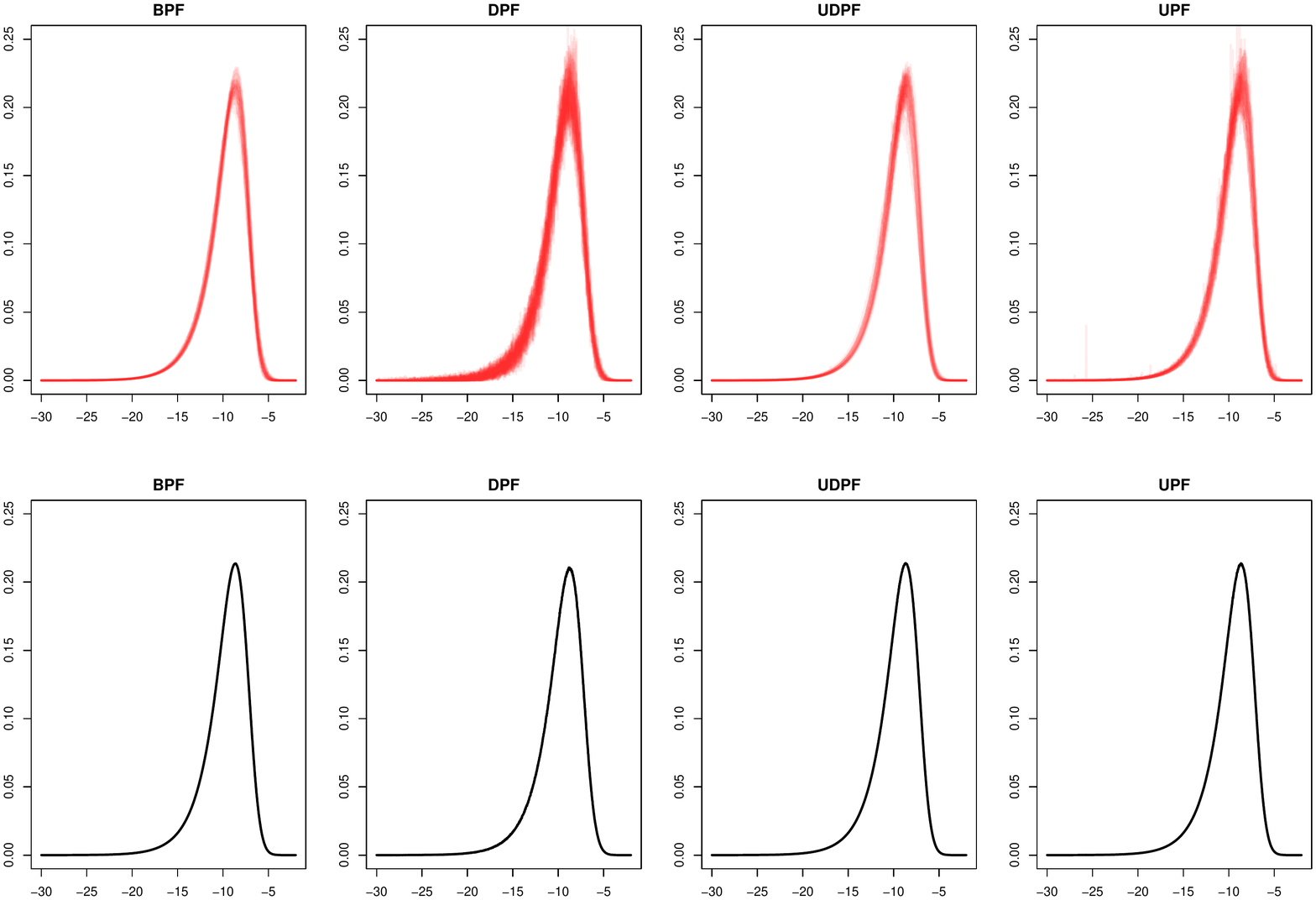}
\end{center}
\caption{The top panel shows all of the individual \textit{conditional
}one-step ahead forecast distributions of $\log(y_{501}^{2})$ produced by BPF
(in black), UPF (in green), DPF (in red) and UDPF (in blue), for the
incorrectly specified SV model, using simulated data. The bottom panel
displays, for each filter, the \textit{marginal} one-step ahead forecast
distribution, which is obtained using the average of the corresponding
conditional forecasts shown in the top panel. }%
\label{forecast_dist_jump_fix}%
\end{figure}

\section{Empirical Illustration\label{emp_section}}

In this section we consider the production of forecast distributions of
log-squared returns from an SV model for daily S\&P500 returns, based on data
from April 6, 2016 to April 2, 2019. The time series plot of the $754$
observations from the sample period, shown in Figure \ref{emp_SP500_data_fix},
suggests that a reasonably sophisticated model such as that in (\ref{svij1})
and (\ref{svij2}) may be appropriate. However, given the robustness of the
forecasts to model mis-specification (as documented above), we use the simpler
SV model in (\ref{SV_meas}) and (\ref{SV_state}) to produce the
forecasts.\ The prior distribution is the same as that detailed in Section
\ref{sec_SV}, and the forecasting performance of each of the filters is
produced using the PMMH procedure described in Section \ref{forecast perf},
with each filter implemented using $N=300$ particles.

The first $T=500$ observations are used to produce an initial one-step-ahead
marginal predictive distribution, corresponding to day $T+1=501$ (April 2,
2018). This process is then repeated for each subsequent period, using an
expanding in-sample window and resulting in a total of $254$ one-step-ahead
predictive distributions. With each out-of-sample predictive and corresponding
to each filter, a log score is produced. The corresponding forecast
performance for the predictive distributions produced using the different
filtering methods, as measured by ALS and ADLS, are reported in Table
\ref{average_ls emp}. These results show that the empirical forecast accuracy
yielded by the distinct filters is almost identical, confirming the robustness
of forecast performance to filter type documented above using simulation.

\begin{figure}[h]
\begin{center}
\includegraphics[scale=0.45]
{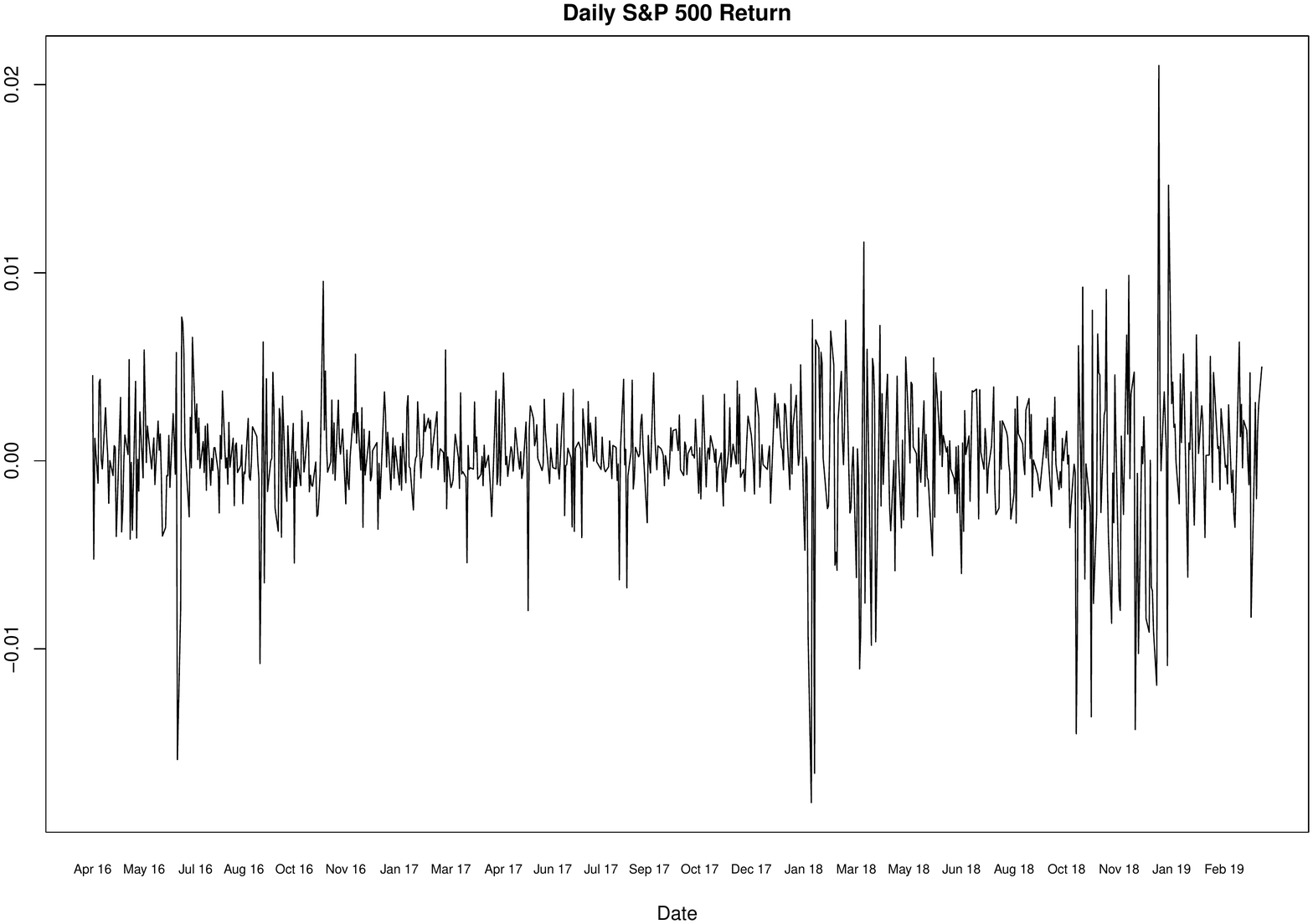}
\end{center}
\caption{Time series plot of the 754 daily S\&P 500 return data from
6-Apr-2016 to 2-Apr-2019.}%
\label{emp_SP500_data_fix}%
\end{figure}%

\begin{table}[tbp] \centering
\caption{ALS (column 2) and ADLS (column 3) resulting from 254 one-step ahead out of sample forecasts obtained from the SV model estimated using PMMH, with four different filters.}\label{average_ls emp}%
\bigskip%
\begin{tabular}
[c]{llll}\hline\hline
&  &  & \\
\multicolumn{4}{c}{Log score summaries}\\\hline
&  &  & \\
&  & ALS & ADLS\\
DPF &  & -2.2447 & 0.0342\\
BPF &  & -2.2389 & 0\\
UPF &  & -2.2471 & 0.0231\\
UDPF &  & -2.2381 & 0.0042\\\hline\hline
\end{tabular}%
\end{table}%

\section{Concluding remarks\label{conclude}}

This paper assesses the impact of filter choice on forecast accuracy in state
space models in which PMMH algorithms are used to estimate the predictive
distribution. To broaden the scope of the investigation, we complement certain
well-established particle filters with two new particle filtering algorithms,
namely the data-driven particle filter (DPF) and the unscented data-driven
particle filter (UDPF), each of which is shown to yield an unbiased estimator
of the likelihood function. In the context of several simulation studies and
an empirical illustration, we show that out-of-sample forecast performance is
largely insensitive to the choice of filtering method employed within the PMMH
algorithm, and for the purpose of undertaking the forward filtering step. This
finding holds irrespective of the correctness or otherwise of the model
specification, including in an empirical setting, where it is inevitable that
some degree of misspecification will prevail.

Of course, if parameter inference were the main focus, the choice of the
filter may matter, with the adapted UDPF shown to perform well across a range
of SNR settings (relative to the BPF, DPF and UPF), and to be on par with the
fully adapted auxiliary particle filter (FAPF) when the latter is available.
However, these and other caveats that may apply to the use of particle
filtering within a Bayesian \textit{inferential} algorithm, do not appear to
apply to forecast performance. Our results suggest that when it comes to
forecasting, the most appropriate decision is to use the simplest filter for
the model at hand, because, in the end: \textquotedblleft Any filter will
do!\textquotedblright.

\appendix\numberwithin{equation}{section}

\section{Appendices}

\subsection{Algorithms for implementing the DPF and UDPF}%

\begin{algorithm}{{\bf1}
The DPF with a pre-specified number of matches $L$, with $1 \leq L \leq N.$}{
\label{algo:DPF_algo}}
Generate ${{x}_{0}^{[j]}}%
$ from the initial state distribution $p(x_0)$, for $j=1,2,...,N$.\\
Set the normalized particle weights ${{\pi}_{0}^{[j]}={\frac{1}{N}}%
}, for j=1,2,...,N$.\\
\qfor$t=0,1,...,T-1:$\\
\hspace*{-0.3in}  Generate ${\eta}_{t+1}^{[j]} \stackrel{i.i.d.}{\sim}
p(\mathbf{\eta}_{t+1})$, for $j=1,2,...,N$.\\
\hspace*{-0.3in} Calculate the particle weights $w_{t+1}^{[j]}%
$ according to (\ref{Matched_DPF_weight}%
), for $j=1,2,...,N$. Note that when $L=1$ this is equivalent to (\ref
{DPF_weight}).\\
\hspace*{-0.3in} Calculate $\hat{p}_{u}(y_{t+1}|y_{1:t},\theta)$ using (\ref
{PF_likelihood}).\\
\hspace*{-0.3in} Calculate the normalized particle weights ${\pi_{t+1}^{[j]}%
}=\frac{{w_{t+1}^{[j]}}}{\sum_{k=1}^{N}{w_{t+1}^{[k]}}}$, for $j=1,2,...,N$.\\
\hspace*{-0.3in} Resample $N$ particles $ x_{t+1}^{[j]}%
$, using probabilities $\{\pi_{t+1}^{[k]}, k=1,2,...,N\}$.\\
\hspace*{-0.3in} Set $\pi_{t+1}^{[j]}={\frac{1}{N}}$.\qrof\end{algorithm}%
%

\begin{algorithm}{ {\bf2} The UDPF.}{
\label{algo:UDPF_algo}}
Generate ${{x}_{0}^{[j]}}%
$ from the initial state distribution $p(x_0)$, for $j=1,2,...,N$.\\
Set the normalized particle weights ${{\pi}_{0}^{[j]}={\frac{1}{N}}%
}, j=1,2,..,N$.\\
\qfor$t=0,1,...,T-1:$\\
\hspace*{-0.3in} Calculate $\hat{\mu}_{M,t+1}$ and $\hat{\sigma}_{M,t+1}%
^{2}$ according to (\ref{unscented_mean}) and (\ref{unscented_variance}).\\
\hspace*{-0.3in} Construct the UDPF proposal distribution using (\ref
{udpf_proposal}), for $j=1,2,...,N$.\\
\hspace*{-0.3in} Generate $x_{t+1}^{[j]}%
$ from the proposal distribution in Step 5, for $j=1,2,...,N.$\\
\hspace*{-0.3in} Calculate the particle weights $w_{t+1}^{[j]}%
, j=1,2,..,N$ according to (\ref{UDPF_weight}).\\
\hspace*{-0.3in} Calculate $\hat{p}_{u}(y_{t+1}|y_{1:t},\theta)$ using (\ref
{PF_likelihood}).\\
\hspace*{-0.3in} Calculate the normalized particle weights ${\pi_{t+1}^{[j]}%
}=\frac{{w_{t+1}^{[j]}}}{\sum_{i=1}^{N}{w_{t+1}^{[j]}}}$, for $j=1,2,...,N$.\\
\hspace*{-0.3in} Resample $N$ particles $x_{t+1}^{[j]}%
$, using probabilities $\{\pi_{t+1}^{[k]}, k=1,2,...,N\}$.\\
\hspace*{-0.3in} Set $\pi_{t+1}^{[j]}={\frac{1}{N}}$.\qrof\end{algorithm}%

\subsection{The use of the unscented transformations in the UDPF}

An unscented transformation is a quick and accurate procedure for calculating
the moments of a non-linear transformation of an underlying random variable.
The procedure involves choosing a set of points, called \textit{sigma points},
from the support of the underlying random variable. Once selected, these sigma
points are weighted to ensure that the first $M-1$ moments of\textbf{ }the
discrete\textbf{ }sigma point distribution\textbf{ }equal the first\textbf{
}$M-1$\textbf{ }moments of\textbf{ }the corresponding distribution of the
underlying random variable. The set of sigma points is then propagated through
the relevant non-linear function, from which the mean and variance of the
resulting normal approximation are obtained. The implied\textbf{ }moments
associated with the weighted transformed points can be shown to match the true
moments of the transformed underlying\textbf{ }random variable up to a
predetermined order of accuracy. (See Julier, Uhlmann and Durrant-Whyte, 1995, 2000.)

In the UDPF, the unscented transformation is applied to the function defined
by solving the measurement equation in (\ref{meas_fcn}) for the state variable
$x_{t+1}$.\textbf{ }We denote $\mu_{\eta}$ and $\sigma_{\eta}^{2}$
respectively as the expected value and the variance of the measurement error
$\eta_{t+1}.$ To calculate the mean and variance of the normal approximation
in (\ref{DUPF_likelihood}), sigma points $\eta^{\{k\}}$, with $k=1,...,M$, are
chosen to span the support of $\eta_{t+1}.$ The corresponding weights for each
each sigma point, $Q^{\{k\}}$, are determined to ensure that the first $M-1$
moments of the (discrete) distribution associated with the weighted sigma
points match the corresponding theoretical moments of the underlying
distribution, $p(\eta_{t+1})$.\emph{ }Accordingly, the sigma point\textbf{
}weights satisfy the following system of equations%
\[
\left\{
\begin{array}
[c]{c}%
\sum_{k=1}^{M}Q^{\{k\}}\\
\sum_{k=1}^{M}Q^{\{k\}}(\eta^{\{k\}}-\mu_{\eta})\\
\vdots\\
\sum_{k=1}^{M}Q^{\{k\}}(\eta^{\{k\}}-\mu_{\eta})^{M-1}%
\end{array}
=%
\begin{array}
[c]{c}%
1\\
E\left[  \eta_{t+1}-\mu_{\eta}\right] \\
\vdots\\
E\left[  (\eta_{t+1}-\mu_{\eta})^{M-1}\right]
\end{array}
\right.  .
\]
Note that, if the measurement errors have the same distribution for all $t$,
then the weighted sigma point distribution will also be the same for all $t$,
and hence will require calculation only once. This is the situation for all
models considered in the paper.

For implementation of the unscented transformations within the UPDF, let the
mean of the distribution whose density is proportional to the measurement
density be given by%
\[
\mu_{M,t+1}=\int_{-\infty}^{\infty}x_{t+1}C_{t+1}p(y_{t+1}|x_{t+1}%
,\theta)dx_{t+1},
\]
where $C_{t+1}=(%
{\textstyle\int}
p(y_{t+1}|x_{t+1},\theta)dx_{t+1})^{-1}$ represents the normalizing constant
that ensures a proper density. Further, using the representation of the
measurement density in (\ref{Dirac_inversion}),\textbf{ }we have%
\begin{equation}
\mu_{M,t+1}=\int_{-\infty}^{\infty}x_{t+1}C_{t+1}\int_{-\infty}^{\infty}%
p(\eta_{t+1})\left\vert \frac{\partial h}{\partial x_{t+1}}\right\vert
_{x_{t+1}=x_{t+1}(y_{t+1},\eta_{t+1})}^{-1}\delta_{x_{t+1}(y_{t+1},\eta
_{t+1})}d\eta_{t+1}dx_{t+1}.\nonumber
\end{equation}
Using then the discrete approximation of $p\left(  \eta_{t+1}\right)  $
implied by the weighted sigma points,\textbf{ }$\eta^{\{k\}}$\textbf{
}for\textbf{ }$k=1,2,...,M$,\emph{ }the mean of the measurement component as
calculated by the unscented transformation satisfies \emph{ }%
\begin{align}
\widehat{\mu}_{M,t+1}  &  =\int_{-\infty}^{\infty}x_{t+1}C_{t+1}\int_{-\infty
}^{\infty}\widehat{p}(\eta_{t+1})\left\vert \frac{\partial h}{\partial
x_{t+1}}\right\vert _{x_{t+1}=x_{t+1}(y_{t+1},\eta_{t+1})}^{-1}\delta
_{x_{t+1}(y_{t+1},\eta_{t+1})}d\eta_{t+1}dx_{t+1}\nonumber\\
&  =\int_{-\infty}^{\infty}x_{t+1}C_{t+1}\int_{-\infty}^{\infty}\left[
{\textstyle\sum\limits_{k=1}^{M}}
Q^{\{k\}}\delta_{\eta^{\{k\}}}\right]  \left\vert \frac{\partial h}{\partial
x_{t+1}}\right\vert _{x_{t+1}=x_{t+1}(y_{t+1},\eta_{t+1})}^{-1}\delta
_{x_{t+1}(y_{t+1},\eta_{t+1})}d\eta_{t+1}dx_{t+1}\nonumber\\
&  =\frac{%
{\textstyle\sum\limits_{k=1}^{M}}
Q^{\{k\}}\left\vert \frac{\partial h}{\partial x_{t+1}}\right\vert
_{\eta_{t+1}=\eta^{\{k\}}\text{, }x_{t+1}=x_{t+1}(y_{t+1},\eta^{\{k\}})}%
^{-1}x_{t+1}(y_{t+1},\eta^{\{k\}})}{%
{\textstyle\sum\limits_{j=1}^{M}}
Q^{\{j\}}\left\vert \frac{\partial h}{\partial x_{t+1}}\right\vert
_{\eta_{t+1}=\eta^{\{j\}}\text{, }x_{t+1}=x_{t+1}(y_{t+1},\eta^{\{j\}})}^{-1}%
}. \label{unscented_mean}%
\end{align}
Similarly, the variance of the measurement component as calculated by the
unscented transformation is given by%
\begin{equation}
\widehat{\sigma}_{M,t+1}^{2}=\frac{%
{\textstyle\sum\limits_{k=1}^{M}}
Q^{\{k\}}\left\vert \frac{\partial h}{\partial x_{t+1}}\right\vert
_{\eta_{t+1}=\eta^{\{k\}}\text{, }x_{t+1}=x_{t+1}(y_{t+1},\eta^{\{k\}})}%
^{-1}(x_{t+1}(y_{t+1},\eta^{\{k\}})-\widehat{\mu}_{M,t+1})^{2}}{%
{\textstyle\sum\limits_{j=1}^{M}}
Q^{\{k\}}\left\vert \frac{\partial h}{\partial x_{t+1}}\right\vert
_{\eta_{t+1}=\eta^{\{j\}}\text{, }x_{t+1}=x_{t+1}(y_{t+1},\eta^{\{j\}})}^{-1}%
}. \label{unscented_variance}%
\end{equation}

\subsection{The unbiasedness of the data-driven filters\label{pmcmc}}

As discussed in Section \ref{like_est}, the unbiasedness condition in
(\ref{unbiased}) is required to ensure that\textbf{ }a\textbf{ }PMMH scheme
yields the correct invariant posterior distribution for $\theta$, \textbf{a}nd
so we consider the theoretical properties of the new filters proposed in
Sections \ref{Section_DPF} and \ref{Section_UDPF} here.\textbf{ }While the
proof in Pitt \textit{et al.} (2012) demonstrates the unbiasedness property of
the likelihood estimator produced by the APF, their proof also applies for the
BPF and UPF. That the proof applies to the BPF is noted by Pitt \textit{et
al}., with the critical insight being that the first step of their Algorithm 1
has no impact, so that each previous particle $x_{t}^{[j]}$ retains the weight
of $\pi_{t}^{[j]}=1/N$. This is also true of the UPF, as the information
regarding the next observation $y_{t+1}$ is incorporated into the proposal
distribution at time $t+1$ through an unscented transformation, and not via an
additional resampling step.

However, due to the multiple matching technique, which is only available for
use with the DPF (and not with either\textbf{ }the APF or the UPF), the proof
in Pitt \textit{et al.} is not adequate to prove Theorem 1. In this case we
provide all details of the proof of Theorem 1 here, along with those of two
lemmas upon which our proof depends. In particular, the following two theorems
establish that the unbiasedness condition holds for all versions of the
data-driven filter, namely the DPF with single ($L=1$), partial ($L<N$) or
full ($L=N$) matching. The collective\textbf{ }conditions C1 - C3 detailed
below, which ensure that the outlined algorithms produce well-defined proposal
distributions, are assumed when deriving the unbiasedness of the resulting
likelihood estimators.

\begin{enumerate}
\item[C1.] For each fixed value $x$, the function $h(x,\eta)\ $is a strictly
monotone function of $\eta$, with continuous non-zero (partial) derivative.

\item[C2.] For each fixed value $y$, the function $x(y,\eta),\ $defined
implicitly by $y=h\left(  x,\eta\right)  $, is a strictly monotone function of
$\eta$, with continuous non-zero (partial) derivative.

\item[C3.] The conditions\ $\int x_{t+1}^{k}p(y_{t+1}|x_{t+1},\theta
)dx_{t+1}<\infty$\ hold,\ for\textbf{ }$k=0$,\textbf{ }$1$,\textbf{
}and\textbf{ }$2$.
\end{enumerate}

\begin{theorem}
\label{theorem1}Under C1\ through C2, any likelihood estimator produced by a
DPF is unbiased. That is,\textbf{ }the likelihood estimator\textbf{
}$\widehat{p}_{u}(y_{1:T}|\theta)$ resulting from any such filter,\textbf{
}with $1\leq L\leq N$\textbf{ }matches,\textbf{ }satisfies%
\[
E_{u}[\widehat{p}_{u}(y_{1:T}|\theta)]=p(y_{1:T}|\theta)\text{.}%
\]

\end{theorem}

We adapt the proof from Pitt\textit{ \textit{et al.}} (2012) to demonstrate
the unbiasedness of the new likelihood estimators specified under Theorem 1,
and represented generically by
\begin{equation}
\widehat{p}_{u}(y_{1:T}|\theta)=\widehat{p}_{u}(y_{1}|\theta)\prod_{t=2}%
^{T}\widehat{p}_{u}(y_{t}|y_{1:t-1},\theta), \label{like_estimator}%
\end{equation}
where unbiasedness means that\textbf{ }$E[\widehat{p}_{u}(y_{1:T}%
)|\theta]=p(y_{1:T}|\theta)$. The factors in (\ref{like_estimator}) are given
in (\ref{PF_likelihood}) for each\textbf{ }$t=1,2,...,T$, with the
weights\textbf{ }$w_{t+1}^{[j]}$\textbf{ }defined by the relevant choice of
$L$. As conditioning on the parameter $\theta$ remains in all subsequent
expressions, we again suppress\textbf{ }its explicit inclusion to help\textbf{
}simplify the expressions\textbf{ }throughout the remainder of this appendix.

Firstly, as noted above, the conditions outlined for this theorem ensure that
the proposal distribution,\textbf{ }$g(x_{t+1}|y_{t+1},\theta)$, is well
defined. Next, let $u$ denote the vector of canonical $i.i.d.$ random
variables used to implement the given filtering algorithm, and let $F_{t}$ be
the subset of such variables generated\textbf{ }up to and including time
$t$,\textbf{ }for each\textbf{ }$t=0,1,...,T$. This means that by conditioning
on $F_{t}$, the particle set $\left\{  x_{0:t}^{[1]},x_{0:t}^{[2]}%
,...,x_{0:t}^{[N]}\right\}  $ and the associated normalized weights $\left\{
\pi_{t}^{[1]},\pi_{t}^{[2]},...,\pi_{t}^{[N]}\right\}  \ $that together
provide the approximation of the filtered density, as in (\ref{phat_xtplus1}),
are assumed to be known. Following\textbf{ }Pitt\textit{ \textit{et al.}
}(2012), in order to prove the unbiasedness property of the likelihood
estimator we require the following two lemmas:

\begin{lemma}
\label{lem1}%
\[
E_{u}[\widehat{p}_{u}(y_{T}|y_{1:T-1},\theta)|F_{T-1}]=\sum_{j=1}^{N}\pi
_{T-1}^{[j]}p(y_{T}|x_{T-1}^{[j]},\theta).
\]

\end{lemma}

\begin{lemma}
\label{lem2}%
\begin{equation}
E_{u}[\widehat{p}_{u}(y_{T-h:T}|y_{1:T-h-1},\theta)|F_{T-h-1}]=\sum_{j=1}%
^{N}\pi_{T-h-1}^{[j]}p(y_{T-h:T}|x_{T-h-1}^{[j]},\theta). \label{lemma2}%
\end{equation}

\end{lemma}

According to Section \ref{Section_DPF}, the estimator of the likelihood
component for the DPF (with potential multiple matching), for given\textbf{
}$1\leq L\leq N,$ is%
\begin{align}
\widehat{p}_{u}(y_{t+1}|y_{1:t},\theta)  &  =\sum_{j=1}^{N}w_{t+1}%
^{[j]}\nonumber\\
&  =\sum_{j=1}^{N}\left(  \frac{1}{L}\sum_{l=1}^{L}w_{t+1}^{[j][l]}\right)
\nonumber\\
&  =\sum_{j=1}^{N}\left(  \frac{1}{L}\sum_{l=1}^{L}\frac{p(y_{t+1}%
|x_{t+1}^{[j]},\theta)\pi_{t}^{[k_{l,j}]}p(x_{t+1}^{[j]}|x_{t}^{[k_{l,j}%
]},\theta)}{g(x_{t+1}^{[j]}|y_{t+1},\theta)}\right)  , \label{LE}%
\end{align}
where the proposal distribution is given in (\ref{DPF_proposal_dist}) and
$k_{l,j}$ represents the $j^{th}$ component of the $l^{th}$ cyclic
permutation, $K_{l}$, as defined in Section \ref{Section_DPF}.

\bigskip

\begin{proof}
[Proof of Lemma \ref{lem1}]We start with%
\begin{align*}
E_{u}[\widehat{p}_{u}(y_{T}|y_{1:T-1},\theta)|F_{T-1}]  &  =E_{u}\left[
\left.  \sum_{j=1}^{N}\left(  \frac{1}{L}\sum_{l=1}^{L}\frac{p(y_{T}%
|x_{T}^{[j]},\theta)\pi_{T-1}^{[k_{l,j}]}p(x_{T}^{[j]}|x_{T-1}^{[k_{l,j}%
]},\theta)}{g(x_{T}^{[j]}|y_{T},\theta)}\right)  \right\vert F_{T-1}\right] \\
&  =\frac{1}{L}\sum_{j=1}^{N}\sum_{l=1}^{L}E_{u}\left[  \left.  \pi
_{T-1}^{[k_{l,j}]}\frac{p(x_{T}^{[j]}|x_{T-1}^{[k_{l,j}]},\theta)p(y_{T}%
|x_{T}^{[j]},\theta)}{g(x_{T}^{[j]}|y_{T},\theta)}\right\vert F_{T-1}\right]
.
\end{align*}
The randomness of each component within the double\textbf{ }summation, for
which the expectation is to be taken, comes from the proposal distribution
that simulates the particle\emph{ }$x_{T}^{[j]}$. Hence, the expectation can
be replaced with its integral form explicitly as:\emph{ }%
\begin{align}
E_{u}[\widehat{p}_{u}(y_{T}|y_{1:T-1},\theta)|F_{T-1}]  &  =\frac{1}{L}%
\sum_{j=1}^{N}\sum_{l=1}^{L}\int\pi_{T-1}^{([k_{l,j}]}\frac{p(x_{T}%
|x_{T-1}^{[k_{l,j}]},\theta)p(y_{T}|x_{T},\theta)}{g(x_{T}|y_{T},\theta
)}g(x_{T}|y_{T},\theta)dx_{T}\label{step_1}\\
&  =\frac{1}{L}\sum_{j=1}^{N}\sum_{l=1}^{L}\left\{  \pi_{T-1}^{[k_{l,j}]}\int
p(y_{T},x_{T}|x_{T-1}^{[k_{l,j}]},\theta)dx_{T}\right\} \nonumber\\
&  =\frac{1}{L}\sum_{j=1}^{N}\sum_{l=1}^{L}\pi_{T-1}^{[k_{l,j}]}%
p(y_{T}|x_{T-1}^{[k_{l,j}]},\theta).\nonumber
\end{align}
Since the $N$ permutations of the$\ $previous particles are mutually
exclusive, each of the terms within the double summation appears exactly $L$
times. Therefore,%
\begin{align*}
E_{u}[\widehat{p}_{u}(y_{T}|y_{1:T-1},\theta)|F_{T-1}]  &  =\frac{1}{L}%
\sum_{j=1}^{N}L\left[  \pi_{T-1}^{[j]}p(y_{T}|x_{T-1}^{[j]},\theta)\right] \\
&  =\sum_{j=1}^{N}\pi_{T-1}^{[j]}p(y_{T}|x_{T-1}^{[j]},\theta).
\end{align*}
Hence, Lemma \ref{lem1} holds.\bigskip
\end{proof}

\begin{proof}
[Proof of Lemma \ref{lem2}]To prove Lemma \ref{lem2}, we use method of
induction as per Pitt\textit{ \textit{et al.}} \textbf{(}2012\textbf{) }First
note that, according to Lemma \ref{lem1}, (\ref{lemma2}) holds when $h=0$.
Next, assuming that Lemma \ref{lem2} holds for any integer\textbf{ }$h\geq0$,
we show that it also holds for $h+1$.

By the law of iterated expectations, we have%
\begin{align*}
&  E_{u}[\widehat{p}_{u}(y_{T-h-1:T}|y_{1:T-h-2},\theta)|F_{T-h-2}]\\
&  =E_{u}\left[  E_{u}\left[  \widehat{p}_{u}(y_{T-h:T}|y_{1:T-h-1}%
,\theta)|F_{T-h-1}\right]  \widehat{p}_{u}(y_{T-h-1}|y_{1:T-h-2}%
,\theta)|F_{T-h-2}\right]  .
\end{align*}
By substituting the formula of $\widehat{p}_{u}(y_{T-h-1}|\theta,y_{1:T-h-2})$
and using the assumption that Lemma \ref{lem2} holds for $h$, we have%
\begin{align*}
&  E_{u}[\widehat{p}_{u}(y_{T-h-1:T}|y_{1:T-h-2},\theta)|F_{T-h-2}]\\
&  =E_{u}\left[  \left.  \left\{  \sum_{j=1}^{N}\pi_{T-h-1}^{[j]}%
p(y_{T-h:T}|x_{T-h-1}^{[j]},\theta)\right\}  \left\{  \sum_{j=1}^{N}%
w_{T-h-1}^{[j]}\right\}  \right\vert F_{T-h-2}\right]
\end{align*}
and noting that $\pi_{T-h-1}^{[j]}$ is the normalized version of
$w_{T-h-1}^{[j]}$, then%
\begin{align*}
&  E_{u}[\widehat{p}_{u}(y_{T-h-1:T}|y_{1:T-h-2},\theta)|F_{T-h-2}]\\
&  =E_{u}\left[  \left.  \left\{  \frac{\sum_{j=1}^{N}p(y_{T-h:T}%
|x_{T-h-1}^{[j]},\theta)w_{T-h-1}^{[j]}}{\sum_{k=1}^{N}w_{T-h-1}^{[k]}%
}\right\}  \left\{  \sum_{j=1}^{N}w_{T-h-1}^{[j]}\right\}  \right\vert
F_{T-h-2}\right] \\
&  =\sum_{j=1}^{N}E_{u}\left[  \left.  p(y_{T-h:T}|x_{T-h-1}^{[j]}%
,\theta)w_{T-h-1}^{[j]}\right\vert F_{T-h-2}\right]  .
\end{align*}
Adopting a similar procedure to that above, owing to the fact that the
expectation is taken with respect to the relevant proposal distribution and
that the multiple matches employ only cyclic rotations, we have%
\begin{align*}
&  E_{u}\left[  \left.  \widehat{p}_{u}(y_{T-h-1:T}|y_{1:T-h-2},\theta
)\right\vert F_{T-h-2}\right] \\
&  =\sum_{j=1}^{N}E_{u}\left[  \left.  p(y_{T-h:T}|x_{T-h-1}^{[j]}%
,\theta)\frac{p(y_{T-h-1}|x_{T-h-1}^{[j]},\theta)\frac{1}{L}\sum_{l=1}^{L}%
\pi_{T-h-2}^{[k_{l,j}]}p(x_{T-h-1}^{[j]}|x_{T-h-2}^{[k_{l,j}]},\theta
)}{g(x_{T-h-1}^{[j]}|y_{T-h-1},\theta)}\right\vert F_{T-h-2}\right] \\
&  =\frac{1}{L}\sum_{j=1}^{N}\sum_{l=1}^{L}\pi_{T-h-2}^{[k_{l,j}]}\int
p(y_{T-h:T}|x_{T-h-1},\theta)p(y_{T-h-1}|x_{T-h-1},\theta)p(x_{T-h-1}%
|x_{T-h-2}^{[k_{l,j}]},\theta)dx_{T-h-1}\\
&  =\sum_{j=1}^{N}\left\{  \pi_{T-h-2}^{[j]}\int p(y_{T-h:T}|x_{T-h-1}%
,\theta)p(y_{T-h-1}|x_{T-h-1},\theta)p(x_{T-h-1}|x_{T-h-2}^{[j]}%
,\theta)dx_{T-h-1}\right\} \\
&  =\sum_{j=1}^{N}\pi_{T-h-2}^{[j]}p(y_{T-h-1:T}|x_{T-h-2}^{[j]},\theta)
\end{align*}
as required.
\end{proof}

\begin{proof}
[Proof of Theorem \ref{theorem1}]From Lemma \ref{lem2}, when $h=T-1$, then
\[
E_{u}\left[  \left.  \widehat{p}_{u}(y_{1:T}|\theta)\right\vert F_{0}\right]
=\sum_{j=1}^{N}p(y_{1:T}|x_{0}^{[j]},\theta)\pi_{0}^{[j]}.
\]
Next, marginalizing over the randomness of $u$ associated with generating a
set of equally weighted particles, $\left\{  x_{0}^{[1]},x_{0}^{[2]}%
,...,x_{0}^{[N]}\right\}  $ at time $t=0$\emph{ }from the initial distribution
$p(x_{0}),$ we have%
\begin{align*}
E_{u}\left[  \widehat{p}_{u}(y_{1:T}|\theta)\right]   &  =E_{u}\left[
E_{u}\left[  \left.  \widehat{p}_{u}(y_{1:T}|\theta)\right\vert F_{0}\right]
\right] \\
&  =E_{u}\left[  \sum_{j=1}^{N}p(y_{1:T}|x_{0}^{[j]},\theta)\pi_{0}%
^{[j]}\right] \\
&  =\frac{1}{N}\sum_{j=1}^{N}E_{u}\left[  p(y_{1:T}|x_{0}^{[j]},\theta
)\right]  .
\end{align*}
Finally, since\textbf{ }the expectation of $p(y_{1:T}|x_{0}^{[j]},\theta)$ is
the same for all $j$, then%
\begin{align*}
E_{u}\left[  \widehat{p}_{u}(y_{1:T}|\theta)\right]   &  =E_{u}\left[
p(y_{1:T}|x_{0},\theta)\right] \\
&  =\int p(y_{1:T}|x_{0},\theta)p(x_{0}|\theta)dx_{0}\\
&  =p(y_{1:T}|\theta),
\end{align*}
and the unbiasedness property of the likelihood estimator associated with each
of the DPF algorithms specified under Theorem \ref{theorem1} is established.
\end{proof}

\bigskip

\begin{theorem}
\label{theorem2}Under C1 through C3, the likelihood estimator produced by the
UDPF filter is unbiased.\emph{ }
\end{theorem}

\begin{proof}
[Proof of Theorem \ref{theorem2}]By recognizing the similarity between the
UDPF\ and the APF, the proof of Theorem \ref{theorem2} can be deduced
directly\textbf{ }from the unbiasedness proof of Pitt\textit{ \textit{et al.}}
(2012). In reference to Algorithm 1 of Pitt\textit{ \textit{et al.}}, the UDPF
algorithm can be reconstructed by setting $g(y_{t+1}|x_{t}^{[j]},\theta)=1$
and with\textbf{ }the proposal distribution, $g(x_{t+1}|x_{t}^{[j]}%
,y_{t+1},\theta)$, formed as per (\ref{udpf_proposal}). All that is required
is that we ensure, through sufficient conditions C1 - C3, that the approximate
moments $\widehat{\mu}_{M,t+1}$ and $\widehat{\sigma}_{M,t+1}^{2}$ used to
obtain this Gaussian proposal distribution are finite.
\end{proof}

\end{document}